\documentclass[conference]{IEEEtran}

\usepackage[cmex10]{amsmath}
\usepackage[pdftex]{graphicx}
\usepackage{latexsym}
\usepackage{amssymb}
\usepackage{bm}
\usepackage{psfrag}
\usepackage{mathdots}
\usepackage{amsthm}
\usepackage{multirow}
\usepackage{subfigure}

\usepackage[pdftex]{graphicx,xcolor}
\usepackage{cite}

\theoremstyle{definition}

\newtheorem{theorem}{Theorem}

\newtheorem{def.}{Definition}
\newtheorem{th.}{Theorem}
\newtheorem{le.}{Lemma}
\newtheorem{ex.}{Example}
\newtheorem{re.}{Remark}

\renewcommand{\vec}[1]{\boldsymbol{#1}}
\newcommand{\mat}[1]{\mathbf{#1}}
\newcommand{\rank}{\mathrm{rank}}
\newcommand{\wt}{\mathrm{wt}}

\begin{document}

\sloppy

\title{
  Efficient Encoding Algorithm of Binary and Non-Binary LDPC Codes Using Block Triangulation
}

\author{
\IEEEauthorblockN{Yuta Iketo and Takayuki Nozaki}
\IEEEauthorblockA{
Dept. of Informatics, Yamaguchi University, JAPAN\\
Email: \{b001vb,tnozaki\}@yamaguchi-u.ac.jp
}
}
\maketitle

\begin{abstract}
  We propose an efficient encoding algorithm for the binary and non-binary low-density parity-check codes.
  This algorithm transforms the parity part of the parity-check matrix into a block triangular matrix with low weight diagonal submatrices by row and column permutations in the preprocessing stage.
  This algorithm determines the parity part of a codeword by block back-substitution to reduce the encoding complexity in the encoding stage.
  Numerical experiments show that this algorithm has a lower encoding complexity than existing encoding algorithms.
  Moreover, we show that this algorithm encodes the non-binary cycle codes in linear time.
\end{abstract}

\section{Introduction}\label{sec:1}
Low-density parity-check (LDPC) codes \cite{gallager1962low} are defined by sparse parity-check matrices $\mat{H} \in \mathbb{F}^{m \times n}_q$, where $\mathbb{F}_q$ stands for Galois field of order $q$.
The LDPC codes are decoded by the sum-product decoding algorithm in linear time and come close to the Shannon capacity on many channels \cite{richardson2001design}.
It is known that irregular LDPC codes have lower block-error rates than regular LDPC codes \cite{hu2005regular}, \cite{amraoui2006asymptotic}.

Let $\mat{G} \in \mathbb{F}^{(n-m) \times n}_q$ be a generator matrix corresponding to the parity-check matrix $\mat{H}$.
In general, encoding for linear codes maps a message $\vec{u} \in \mathbb{F}^{n-m}_q$ into a codeword $\vec{x} \in \mathbb{F}^{n}_q$ by $\vec{x}^T = \vec{u}^T\mat{G}$.
Even if $\mat{H}$ is a sparse matrix, $\mat{G}$ is not always a sparse matrix.
Hence, the complexity of an encoding algorithm (EA) by a generator matrix is quadratic in the code length.
In other words, an EA by a generator matrix has a larger complexity than the sum-product decoding algorithm for long code length.
The encoding can be a bottleneck in a communication system with an LDPC code.
Therefore, it is important to develop a low complexity EA for LDPC codes.

In this paper, we assume that the parity-check matrix $\mat{H}$ has full rank,
i.e., $\rank (\mat{H}) = m$. 
By performing suitable column permutations into $\mat{H}$, 
we get $\left( \mat{H}_P \ \ \mat{H}_I \right)$,
where $\mat{H}_P$ is non-singular.
Then the codeword $\vec{x} \in \mathbb{F}^{n}_q$ is split into two parts, namely the parity part $\vec{p} \in \mathbb{F}^{m}_q$ and the message part $\vec{u} \in \mathbb{F}^{n-m}_q$, as $\vec{x}=(\vec{p} \ \vec{u})$.
Since $\mat{H}\vec{x}^{T} = \vec{0}^T$, we get $\mat{H}_P\vec{p}^T = -\mat{H}_I\vec{u}^T$.

In general, the encoding for LDPC codes is accomplished in two stages: the preprocessing stage and the encoding stage.
In the preprocessing stage, we split $\mat{H}$ into $\mat{H}_I$ corresponding to $\vec{u}$, and $\mat{H}_P$ corresponding to $\vec{p}$.
In the encoding stage, we determine $\vec{p}$ by solving $\mat{H}_P\vec{p}^T = -\mat{H}_I\vec{u}^T$.
Since the computation of $\vec{b}^T := -\mat{H}_I\vec{u}^T$ is the product of a sparse matrix and a known vector, the computation complexity is linear in the code length.
Hence, in encoding for LDPC codes, we should consider 
(i) an algorithm to construct suitable $\mat{H}_P$ and (ii) an algorithm to solve $\mat{H}_P\vec{p}^T=\vec{b}^T$ efficiently.
We transform $\mat{H}$ only by row and column permutations to keep the sum-product decoding performance.
Hence, by using two permutation matrices $\mat{P} \in \mathbb{F}^{m \times m}_q$ and $\mat{Q} \in \mathbb{F}^{n \times n}_q$, 
we transform $\mat{H}$ into $\mat{P}\mat{H}\mat{Q}=\left( \mat{H}_P \ \ \mat{H}_I \right)$.

The existing works of the EA are itemized as follows and summarized in Table \ref{tab:Compres}.
\begin{enumerate}
\item Richardson and Urbanke \cite{richardson2001efficient} proposed an efficient EA by transforming $\mat{H}_P$ into an approximate triangular matrix (ATM).
The complexity of this EA is $O(n+\delta^2)$, where $n$ is a code length and $\delta$ is a \textit{gap} that satisfies $\delta \ll n$ and is proportional to $n$.
Hence, the complexity of this EA is $O(n^2)$ but is lower than that of the EA by a generator matrix.

\item Kaji \cite{kaji2006encoding} proposed an EA by the LU-factorization.
  This EA has lower complexity for codes with small gaps than Richardson and Urbanke's EA (RU-EA).

\item Shibuya and Kobayashi \cite{shibuya2014efficient} proposed an EA by block-triangulation.
This EA is more efficient than the RU-EA.
However, this EA has a possibility of abend \cite{nozaki2014},
More precisely, there is a possibility that some diagonal submatrices become vertical, i.e., singular.

\item Nozaki \cite{nozaki2015parallel} proposed a parallel EA by block-diagonalization.
This EA has a lower encoding time than the RU-EA by parallel computation.
However, this EA has a slightly larger total computation complexity than the RU-EA.

\item Huang and Zhu \cite{huang2006linear} proposed a linear time EA for non-binary cycle codes, i.e., codes defined by parity-check matrices of which all the column weights are two.
This EA transforms $\mat{H}_P$ into an upper block-bidiagonal matrix whose diagonal submatrices are cycle or diagonal matrices.
\end{enumerate}

\begin{table*}[!t]
  \begin{center}
    \caption{Comparison with existing works}
    \begin{tabular} {lcclll} \hline 
      \multirow{2}{*}{Work} & \multicolumn{2}{c}{Field of code} & \multirow{2}{*}{Type of code} & \multirow{2}{*}{Technique} & \multirow{2}{*}{Note}\\
      & binary & non-binary & & & \\ \hline
      Richardson and Urbanke \cite{richardson2001efficient} & yes & yes & Irregular & Approximate triangulation & \\
      Kaji \cite{kaji2006encoding} & yes & yes & Irregular & LU factorization &\\
      Shibuya and Kobayashi \cite{shibuya2014efficient} & yes & yes & Irregular & Block triangulation & Possibility of abend \\
      Nozaki \cite{nozaki2015parallel} & yes & yes & Irregular & Block diagonalization & Encodable at parallel\\
      Huang and Zhu \cite{huang2006linear} & no & yes & Cycle & Block bidiagonalization & Encodable in linear \\
      This work & yes & yes & Irregular & Block triangulation & \\ \hline
    \end{tabular}
    \label{tab:Compres}
  \end{center}
\end{table*}

The goal of this work is to reduce the encoding complexity for binary and non-binary irregular LDPC codes.
The main idea of this work is the block triangularization of a given parity check matrix.
In detail, our proposed EA transforms each of the diagonal submatrices into a cycle or a diagonal matrix as much as possible.
For the non-binary cycle codes, the resulting matrix by the proposed EA coincides with one by Huang and Zhu's EA (HZ-EA).
In other words, the proposed EA is a generalization of the HZ-EA.
Therefore, the complexity of the proposed EA is $O(n)$ for the non-binary cycle codes.
Numerical experiments show that the proposed EA has lower complexity than the RU-EA and Kaji's EA (K-EA).

We organize the rest of the paper as follows.
Section \ref{sec:pre} introduces LDPC codes and existing EAs.
Section \ref{sec:3} presents the proposed EA.
Numerical examples in Sect.\ \ref{sec:4} compare the complexity of the proposed EA and existing EAs.

\section{Preliminaries}\label{sec:pre}
This section introduces LDPC codes and existing EAs.

\subsection{LDPC Codes}
LDPC codes are defined by sparse parity-check matrices $\mat{H} \in \mathbb{F}^{m \times n}_q$ as 
$\mathcal{C}_q =\bigl\{ \vec{x} \in \mathbb{F}_q^n \mid\mat{H}\vec{x}^T=\vec{0}^T \bigr\}$.
In particular, $\mathcal{C}_2$ is a binary LDPC code.
When $q>2$, we call $\mathcal{C}_q$ non-binary LDPC code.

Each parity-check matrix $\mat{H}$ is represented as the Tanner graph \cite{wiberg1995codes}, \cite{wiberg1996codes}.
The rows (resp.\ columns) of $\mat{H}$ correspond to the check (resp.\ variable) nodes.
In the Tanner graphs of $(d_c,d_r)$-regular LDPC codes, all the variable (resp.\ check) nodes are of degree $d_c$ (resp.\ $d_r$).
On the other hand, an irregular LDPC code is chosen the degrees of nodes according to two-degree distributions:
the check degree distribution $\rho(x)=\sum_{i}\rho_ix^{i-1}$ and the variable degree distribution $\lambda(x)=\sum_{i}\lambda_ix^{i-1}$,
where $\rho_i$ (resp.\ $\lambda_i$) is the fraction of edges connecting to check (resp.\ variable) nodes of degree $i$.

In particular, LDPC codes satisfying $\lambda(x)=x$ are known as cycle codes \cite{richardson2008modern}.
Note that there are no restrictions to the check degree distribution for the cycle codes.

\subsection{Richardson and Urbanke's Encoding Algorithm \cite{richardson2001efficient}}

The preprocessing stage transforms a given $\mat{H}$ into an ATM $\mat{H}^{\rm RU}$, described as
\begin{equation*}
  \mat{H}^{\rm RU} = \mat{P}\mat{H}\mat{Q} =
  \begin{pmatrix}
    \mat{T} & \mat{S} & \mat{H}^{\rm RU}_{I,1}\\
    \mat{V} & \mat{N} & \mat{H}^{\rm RU}_{I,2}\\
  \end{pmatrix}
  =
  \begin{pmatrix}
    \mat{A} &
    \begin{matrix}
      \mat{H}^{\rm RU}_{I,1}\\\mat{H}^{\rm RU}_{I,2}\\
    \end{matrix}
  \end{pmatrix},
\end{equation*}
where $\mat{T}$ is an $(m-\delta)\times (m-\delta)$ triangular matrix, 
$\mat{S}$, $\mat{V}$, $\mat{N}$, $\mat{H}^{\rm RU}_{I,1}$, and $\mat{H}^{\rm RU}_{I,2}$ are 
$\delta \times(m - \delta)$, $(m - \delta)\times \delta$, $\delta \times \delta$, $(m - \delta)\times(n - m)$, and $\delta \times(n - m)$ matrices, respectively.
We define $\mat{\Phi}:= \mat{N} - \mat{V} \mat{T}^{-1} \mat{S}$.
We compute its inverse matrix $\mat{\Phi}^{-1}$ before the encoding stage.

This preprocessing stage is regarded as an algorithm whose input is a matrix $\mat{H}$ and output is a tuple of matrices $\bigl(\mat{H}^{\rm RU},\mat{P},\mat{Q}\bigr)$.
Therefore, we denote this algorithm by $\mathsf{ATM}(\mat{H}) \to \bigl(\mat{H}^{\rm RU},\mat{P},\mat{Q}\bigr)$ \cite{nozaki2015parallel}.
This notation will be used in Sect.\ \ref{sec:3-2}.

We split the parity part $\vec{p}$ into $(\vec{p}_1, \vec{p}_2)$, where $\vec{p}_1 \in \mathbb{F}^{m-\delta}_q$ and $\vec{p}_{2} \in \mathbb{F}^{\delta}_q$.
Then, the encoding stage executes the following procedure;
\begin{enumerate}
\item Compute $\vec{b}_1^{T} = \mat{H}_{I,1}^{\rm RU}\vec{u}^{T}$ and $\vec{b}_2^{T} = \mat{H}_{I,2}^{\rm RU}\vec{u}^{T}$
\item Derive $\vec{p}_2$ from $\vec{p}_2^{T} = \mat{\Phi}^{-1} (\mat{V} \mat{T}^{-1} \vec{b}_1^{T} - \vec{b}_2^{T})$
\item Solve $\mat{T} \vec{p}_1^{T} = - \mat{S} \vec{p}_2^{T} - \vec{b}_1^{T}$ by backward-substitution
\end{enumerate}

Consider the matrix $\mat{Z}$ over $\mathbb{F}_q$.
Denote the number of non-zero elements in the matrix $\mat{Z}$ by $\wt(\mat{Z})$.
Let $\mathcal{Z}(\mat{Z})$ be the number of rows that include at least one non-zero element in $\mat{Z}$.
Define $\mathcal{S}(\mat{Z}) := \wt(\mat{Z}) - \mathcal{Z}(\mat{Z})$.
Then, $\wt(\mat{Z})$ (resp.\ $\mathcal{S}(\mat{Z})$) gives the number of multiplication (resp.\ addition) over $\mathbb{F}_q$ to calculate the product of the matrix $\mat{Z}$ and a vector.
By using this, the total number of multiplications $\mu_{\rm RU}$ and additions $\alpha_{\rm RU}$ in the encoding stage of the RU-EA are
\begin{align}
  &\mu_{\rm RU} = \wt\bigl(\mat{H}^{\rm RU}_{I}\bigr) + f_m(\mat{A}),
  \quad 
  \alpha_{\rm RU}  =  \mathcal{S}\bigl(\mat{H}^{\rm RU}_{I}\bigr) + f_a(\mat{A}) ,
  \notag \\ 
  &f_m(\mat{A}) := 2\wt(\mat{T}) + \wt(\mat{V}) + \wt(\mat{S}) + \wt\bigl(\mat{\Phi}^{-1}\bigr),
  \label{eq:def-fm} \\ 
  &f_a(\mat{A}) := 2\mathcal{S}(\mat{T}) + \mathcal{S}(\mat{V}) + \mathcal{S}(\mat{S}) + \mathcal{S}\bigl(\mat{\Phi}^{-1}\bigr) + m.
  \label{eq:def-fa}
\end{align}
Note that $f_m(\mat{A})$ and $f_a(\mat{A})$ represent the number of multiplication and addition to solve $\mat{A} \vec{p} = \vec{b}$, respectively.

\subsection{Kaji's Encoding Algorithm \cite{kaji2006encoding}} \label{sec:kaji}

The preprocessing stage of the K-EA transforms a given $\mat{H}$ into an ATM $\mat{H}^{\rm RU}$.
Next, this stage factorizes $\mat{A} = \mat{L}\mat{U}$, where $\mat{L}$ and $\mat{U}$ are lower and upper triangular matrices, respectively.

This encoding stage is accomplished in three steps:
\begin{enumerate}
\item Compute $\vec{b}^T = -\mat{H}^{\rm RU}_{I} \vec{u}^T$
\item Solve $\mat{L}\vec{v}^T = \vec{b}^T$ by forward-substitution.
\item Solve $\mat{U}\vec{p}^T = \vec{b}^T$ by backward-substitution.
\end{enumerate}
The complexity, i.e., the number of multiplications $\mu_{\rm K}$ and additions $\alpha_{\rm K}$, is evaluated as
\begin{align*}
  \mu_{\rm K} &= \wt\bigl(\mat{H}^{\rm RU}_{I}\bigr) + \wt(\mat{L}) + \wt(\mat{U}),\\
 \alpha_{\rm K}  &= \mathcal{S}\bigl(\mat{H}^{\rm RU}_{I}\bigr) + \mathcal{S}(\mat{L}) + \mathcal{S}(\mat{U}).
\end{align*}
It is known that the K-EA has lower complexity than the RU-EA for the codes with small gaps \cite{kaji2006encoding}.

\subsection{Singly Bordered Block-Diagonalization}\label{sec:2-c}
By the singly bordered block-diagonalization \cite{aykanat2004permuting},
a given $\mat{H}$ is transformed into $\mat{H}^{\rm SBBD}$:
\begin{equation*}
  \mat{H}^{\rm{SBBD}} = \mat{P}\mat{H}\mat{Q} =	
  \begin{pmatrix}
    \mat{B}_1 & \mat{O}   & \mat{Z}_1 \\
    \mat{O}    & \mat{B}_2  & \mat{Z}_2 \\
  \end{pmatrix},
\end{equation*}
where $\mat{B}_1$, $\mat{B}_2$, $\mat{Z}_1$, and $\mat{Z}_2$ are $m_1 \times n_1$, $(m-m_1) \times n_2$, $m_1 \times (n-n_1-n_2)$, and $(m-m_1) \times (n-n_1-n_2)$ matrices, respectively.
In the singly bordered block-diagonalization, we can predetermine the numbers of rows $m_1, m_2$.
However, the numbers of columns $n_1, n_2$ depend on the input matrix $\mat{H}$.
Hence, if $m_1$ is small, the submatrix $\mat{B}_1$ becomes vertical, i.e., $m_1 > n_1$.
Thus, to get a square or horizontal submatrix $\mat{B}_1$, we need to adjust the size of $m_1$.

The singly bordered block-diagonalization is regarded as an algorithm whose inputs are a matrix $\mat{H}$ and $m_1$ and output is a tuple of matrices $\bigl(\mat{H}^{\rm SBBD},\mat{P},\mat{Q}\bigr)$.
Therefore, we denote this algorithm by $\mathsf{SBBD}(\mat{H},m_1) \to \bigl(\mat{H}^{\rm SBBD}, \mat{P},\mat{Q}\bigr)$ \cite{nozaki2015parallel}.
This notation will be used in Sect.\ \ref{sec:3-2}.

\subsection{Huang and Zhu's Encoding Algorithm}
This section shows the HZ-EA \cite{huang2006linear} and complexity to solve the equation with a cycle matrix.

\subsubsection{Preprocessing Stage of HZ-EA}\label{sec:2-d1}
The \textit{associated graph} \cite{huang2006linear} gives a graph representation of a matrix, each of whose columns of weight two.
An $m \times n$ matrix is described by an associated graph of $m$ vertices and $n$ edges.
The vertices $\mathsf{v}_i$ and $\mathsf{v}_j$ are connected by the edge $\mathsf{e}_k$ iff the $(i,k)$-entry and the $(j,k)$-entry of the matrix are non-zero.

The {\it cycle matrix}\footnote{Since its associated graph is cycle, we call it cycle matrix.} $\mat{C} \in \mathbb{F}^{k \times k}_q$ has the following form:
\begin{equation*}
\arraycolsep=3pt
  \mat{C}= \left(
  \begin{array}{@{\hskip2pt}ccccc@{\hskip2pt}}
   \gamma_{1} & 0        & 0        & \cdots     & \beta_{k}   \\
    \beta_{1}  & \gamma_{2} & 0        & \cdots     & 0         \\ 
    0        & \beta_{2}  & \gamma_{3} & \cdots     & 0         \\
    \vdots   & \ddots   & \ddots   & \ddots     & \vdots    \\
    0        & \cdots   & 0        & \beta_{k-1} & \gamma_{k} \\
  \end{array}
    \right),
\end{equation*}
where $\beta_{i}$ and $\gamma_{i}$ is a non-zero element over $\mathbb{F}_q$.

\begin{re.}\label{re:cycle}
  Consider a $k \times k$ cycle matrix $\mat{C}$ over $\mathbb{F}_2$.
  Then, the rank of $\mat{C}$ is $k-1$.
  In other words, all cycle matrices over $\mathbb{F}_2$ are always singular.
\end{re.}

The preprocessing stage of the HZ-EA \cite{huang2006linear} transforms $\mat{H}$ into $\mat{H}^{\rm HZ}$:
\begin{equation*}
  \mat{H}^{\rm HZ}
  =
  \left(
  \begin{array}{cccccc}
    \mat{C}_1 & \mat{E}_{1} & \multicolumn{2}{c}{\multirow{2}{*}{{\Large$\mat{O}$}}} & & \\
             & \mat{D}_{2} & \ddots & & \multicolumn{2}{c}{\multirow{2}{*}{$\mat{H}^{\rm HZ}_{I}$}} \\
    \multicolumn{2}{c}{\multirow{2}{*}{{\Large$\mat{O}$}}} & \ddots & \mat{E}_{r-1} & &  \\
    & & & \mat{D}_{r} & & \\
  \end{array}
  \right),
\end{equation*}
where $\mat{C}_{1}$ is an $m_1\times m_1$ cycle matrix,
$\mat{D}_{i}$ is an $m_i \times m_i$ diagonal matrix ($i=2,3,\dots, r$),
and $\mat{E}_{i}$ is $m_{i} \times m_{i+1}$ matrix ($i=1,2,\dots,r-1$),
and $\mat{H}^{\rm HZ}_{I}$ is $m \times (n-m)$ matrix, respectively.
Note that $\textstyle\sum_{i=1}^{r} m_i = m$.

\subsubsection{Solving Algorithm $\mat{C} \vec{\omega}^T = \vec{b}^T$ and Its Complexity}\label{sec:2.4}
We present an algorithm to solve $\mat{C} \vec{\omega}^T = \vec{b}^T$, where $\vec{b} \in \mathbb{F}^{k}_q$ is known.
Note that this algorithm has a smaller complexity than the algorithm in \cite{huang2006linear}, which requires $4k-3$ multiplications and $3k-5$ additions ($k\ge 3$).

To reduce the complexity, predetermine
\begin{align*}
  \epsilon_i &:= \beta_i\gamma_i^{-1} ~~(i=1,2,\dots,k), \quad
  \iota := 1 + \epsilon_1 \epsilon_2 \cdots \epsilon_{k}, \\
  \eta_i &:= \epsilon_{k} \epsilon_1 \epsilon_2 \cdots \epsilon_{i-1} ~~(i=1,2,\dots,k-1).
\end{align*}
We denote substituting $j$ into $i$ by $i \gets j$.
Table \ref{tab:hyou2} shows the algorithm to solve $\mat{C}\vec{\omega}^T=\vec{b}^T$.
The total number of multiplications $\mu_{\rm c}$ and additions $\alpha_{\rm c}$ are
\begin{align}
  &\mu_{\rm c} = 3k-1, &
  &\alpha_{\rm c}  = 2(k-1).
  \label{eq:cyc_comp}
\end{align}

\begin{table}[t]
  \begin{center}
    \caption{Algorithm to solve $\mat{C}\vec{\omega}^T=\vec{b}^T$}
    \begin{tabular}{|l|c|c|}
      \hline
      Operation     & $\mu$  & $\alpha$  \\ \hline
      $z_1 \gets b_1$     & $0$    & $0$  \\ \hline
      $z_i \gets b_{i} +\epsilon _{i-1}z_{i-1}$ $(i=2,...,k)$ & $k-1$ & $k-1$ \\ \hline
      $y_{k} \gets z_{k}\iota^{-1}$  & $1$    & $0$   \\ \hline
      ${\omega}_{k} \gets \gamma_{k}^{-1}y_{k}$    & $1$    & $0$   \\ \hline
      ${\omega}_{i} \gets \gamma^{-1}_i (z_{i} - y_{k}\eta_i)$ $(i=1,...,k-1)$ & $2k-2$ & $k-1$ \\ \hline
    \end{tabular}
    \label{tab:hyou2}
  \end{center}
\end{table}

\section{Proposed Encoding Algorithm}\label{sec:3}
In this section, we propose an EA by block triangularization.
Section \ref{ssec:overview} gives an overview of the proposed EA.
Sections \ref{sec:3-2} and \ref{sec:3-3} present the preprocessing and encoding stages of this EA, respectively.
Section \ref{sec:3-4} evaluates the encoding complexity.
Section \ref{sec:3-5} shows that this EA is a generalization of the HZ-EA.

\subsection{Overview \label{ssec:overview}}
The preprocessing stage of the proposed EA transforms a given parity check matrix $\mat{H}$ into a block-triangular matrix $\mat{H}'$ by row and column permutations, where
\begin{equation}
  \label{eq:H'}
  \mat{H}'=
  \begin{pmatrix}
    \mat{F}_{1} & \mat{K}_{1,2}   & \cdots   & \! \mat{K}_{1,\ell-1} \! & \mat{K}_{1,\ell}    & \mat{H}'_{I,1} \\
    \mat{O}      & \mat{F}_{2}   & \cdots   & \! \mat{K}_{2,\ell-1} \! & \mat{K}_{2,\ell}   & \mat{H}'_{I,2} \\
    \vdots       & \vdots       & \ddots  & \vdots    & \vdots   & \vdots \\
    \mat{O}      & \mat{O}    &\cdots      & \! \mat{F}_{\ell-1}  \! & \mat{K}_{\ell-1,\ell}      & \mat{H}'_{I,\ell-1} \\
    \mat{O}      & \mat{O}    &\cdots      & \! \mat{O} & \mat{F}_{\ell}      & \mat{H}'_{I,\ell} \\
  \end{pmatrix}.
\end{equation}
Here, $\mat{F}_{i}$ is a non-singular $m_i \times m_i$ matrix, $\mat{K}_{i,j}$ and $\mat{H}'_{I,i}$ are of size $m_i \times m_j$ and $m_i \times (n-m)$, respectively.
Note that $\textstyle\sum_{i=1}^{\ell} m_i = m$.
To reduce the encoding complexity, diagonal block $\mat{F}_i$ forms a diagonal matrix as far as possible.

Suppose that the parity part $\vec{p}$ is split into $\ell$ parts $\vec{p}_i \in \mathbb{F}_q^{m_i}$ ($i=1,2,\dots, \ell$) 
as $\vec{p} = (\vec{p}_1, \vec{p}_2,\dots,\vec{p}_{\ell})$.
Then, the encoding stage decides the parity part by solving $\mat{H}'_P \vec{p} = - \mat{H}'_I \vec{u}$ for a given $\vec{u}$.
This equation is efficiently solved by the block backward-substitution (Sect. \ref{sec:3-3}).

In the following sections, we explain block backward-substitution can reduce the encoding complexity by an example (Sect.\ \ref{sssec:bss})
and briefly explain how to construct a block-triangular matrix (Sect.\ \ref{sssec:ext} and \ref{sssec:nsdb}).

\subsubsection{Reason to Make Block Triangular Matrix \label{sssec:bss}}
Assume $\mat{H}_{P}$ is decomposed in the following block matrix
\begin{equation*}
  \mat{H}_P
  =
  \begin{pmatrix}
    \mat{D} & \mat{K}  & \mat{S}_{u} \\
    \mat{O} & \mat{T}  & \mat{S}_{l} \\
    \mat{O} & \mat{V}  & \mat{N} \\
  \end{pmatrix},
\end{equation*}
where $\mat{D}$ is an $m_1 \times m_1$ diagonal matrix, $\mat{T}$ is an $(m_2-\delta)\times (m_2-\delta)$ triangular matrix,
and $\mat{N}$ is a $\delta \times \delta$ matrix.
Since
$\mat{T}_1 = \begin{pmatrix} \mat{D} & \mat{K} \\ \mat{O} & \mat{T} \end{pmatrix}$
is a triangular matrix, $\mat{H}_P$ is regarded as a large ATM, i.e.,
\begin{equation*}
  \left(
  \begin{array}{cc|c}
    \mat{D} & \mat{K} & \mat{S}_u \\
    \mat{O} & \mat{T} & \mat{S}_l \\ \hline
    \mat{O} & \mat{V} & \mat{N} \\
  \end{array}
  \right)
  =
  \begin{pmatrix}
    \mat{T}_1 & \mat{S} \\
    \mat{V}_1 & \mat{N} \\
  \end{pmatrix},
\end{equation*}
where $\begin{pmatrix}\mat{S}_u \\ \mat{S}_l \end{pmatrix} =: \mat{S}$ and
$\begin{pmatrix}\mat{O} & \mat{V}\end{pmatrix} =:  \mat{V}_1$.
In another interpretation,
since $\mat{A}_2 = \begin{pmatrix} \mat{T} & \mat{S}_l \\ \mat{V} & \mat{N} \end{pmatrix}$
is a small ATM, $\mat{H}_{P}$ is regarded as a block triangular matrix, i.e.,
\begin{equation*}
  \left(
  \begin{array}{c|cc}
    \mat{D} & \mat{K} & \mat{S}_{u} \\ \hline
    \mat{O} & \mat{T} & \mat{S}_{l} \\
    \mat{O} & \mat{V} & \mat{N} \\
  \end{array}
  \right)
  =
  \begin{pmatrix}
    \mat{D} & \mat{K}_2 \\
    \mat{O} & \mat{A}_2 \\
  \end{pmatrix},
\end{equation*}
where
$\begin{pmatrix} \mat{K} & \mat{S}_u \end{pmatrix} =: \mat{K}_2$.

Now, we compare the complexity to solve $\mat{H}_P \vec{p}^T = \vec{b}^T$ in the two interpretations above.
To simplify the discussion, we evaluate the complexity by the number of multiplications.
If we regard $\mat{H}_P$ as a large ATM, then the number of multiplications $\mu_1$ to solve $\mat{H}_P \vec{p}^T = \vec{b}^T$ is derived from Eq.~\eqref{eq:def-fm} as
\begin{align}
  \mu_1
  &=
  2( \wt(\mat{D}) + \wt(\mat{K}) + \wt(\mat{T})) + \wt(\mat{V})
  \notag \\
  &~~~~+ \wt(\mat{S}_u) + \wt(\mat{S}_l) + \wt\bigl(\mat{\Phi}_1^{-1}\bigr),
  \label{eq:mu1}
\end{align}
where $\mat{\Phi}_1 := \mat{N}- \mat{V}_1 \mat{T}_1^{-1}\mat{S}$.

If we regard $\mat{H}_P$ as a block-triangular matrix, then we split $\vec{p}$ into $\vec{p}_1 \in \mathbb{F}_q^{m_1}$ and $\vec{p}_2 \in \mathbb{F}_q^{m_2}$ and obtain the following system of linear equations:
\begin{align}
  \mat{A}_2 \vec{p}_2^T &= \vec{b}_2^T,  \label{eq:ex1-1} \\
  \mat{D} \vec{p}_1^T &= \vec{b}_1^T - \vec{K}_2 \vec{p}_2^T, \label{eq:ex1-2} 
\end{align}
where $\vec{b}_1$ (resp.\ $\vec{b}_2$) stands the first $m_1$ (resp.\ last $m_2$) elements of vector $\vec{b}$.
We solve Eq.~\eqref{eq:ex1-1} by the RU-EA's encoding stage and get $\vec{p}_2$.
Substituting $\vec{p}_2$ into Eq.~\eqref{eq:ex1-2}, we have $\vec{p}_1$.
This calculation is called block backward-substitution.
The total number of multiplications $\mu_2$ to derive $\vec{p}_2$ and $\vec{p}_1$ is
\begin{align}
  \mu_2
  &=
  2\wt(\mat{T}) + \wt(\mat{V}) + \wt(\mat{S}_l) + \wt\bigl(\mat{\Phi}_1^{-2}\bigr)
  \notag \\
  &~~~~+ \wt(\mat{K})+\wt(\mat{S}_u) + \wt(\mat{D}),
  \label{eq:mu2}
\end{align}
where $\mat{\Phi}_2 := \mat{N}- \mat{V} \mat{T}^{-1}\mat{S}_l$.

By subtracting Eq.~\eqref{eq:mu2} from Eq.~\eqref{eq:mu1}, we get
\begin{equation*}
  \mu_1 - \mu_2
  =
  \wt(\mat{D}) +  \wt(\mat{K})  + \wt(\mat{\Phi}_1^{-1})  - \wt(\mat{\Phi}_2^{-1}).
\end{equation*}
Since both $\mat{\Phi}_1$ and $\mat{\Phi}_2$ are $\delta \times \delta$ matrices,
we assume $\wt(\mat{\Phi}_1^{-1})  \approx \wt(\mat{\Phi}_2^{-1})$.
From this, we can reduce about $\wt(\mat{D}) +  \wt(\mat{K})$ multiplications by the block backward-substitution.

Summarizing above, if we obtain a block-triangular matrix whose several diagonal-blocks are diagonal matrices, we can reduce the encoding complexity.
In the following section, we briefly explain how to make such a matrix.

\subsubsection{Extraction of Diagonal Matrices \label{sssec:ext}}
Consider the case that $\mat{H}$ does not contain columns of weight 1.
By a method explained in Sect.~\ref{sssec:nsdb}, we extract a non-singular matrix $\mat{F}_1$ from $\mat{H}$ and get
\begin{equation*}
  \mat{H}_1
  =
  \begin{pmatrix}
    \mat{F}_1 & \mat{M}_1 \\
    \mat{O} & \mat{W}_1
  \end{pmatrix},
\end{equation*}  
by performing row and column permutation to $\mat{H}$.

Since $\mat{M}_1\neq \mat{O}$ in general, there is a possibility that $\mat{W}_1$ contains columns of weight 1.
In such case, by rearranging the rows and columns of $\mat{H}_1$, we extract diagonal matrix $\mat{D}_2$ from $\mat{W}_1$ as
\begin{equation*}
  \mat{H}_2
  =
  \begin{pmatrix}
    \mat{F}_1 & \mat{K}_{1,1} & \mat{M}_1 \\
    \mat{O} & \mat{D}_2 & \mat{M}_2 \\
    \mat{O} & \mat{O}  & \mat{W}_2 \\    
  \end{pmatrix}.
\end{equation*}  
If $\mat{W}_1$ does not contain columns of weight 1, we extract a non-singular matrix $\mat{F}_2$ from $\mat{W}_2$ by a similar way to $\mat{F}_1$:
\begin{equation*}
  \mat{H}_2'
  =
  \begin{pmatrix}
    \mat{F}_1 & \mat{K}_{1,1} & \mat{M}_1 \\
    \mat{O} & \mat{F}_2 & \mat{M}_2 \\
    \mat{O} & \mat{O}  & \mat{W}_2 \\    
  \end{pmatrix}.
\end{equation*}
By repeating the above process, finally we obtain the block-triangular matrix.

When $\mat{H}$ contains columns of weight 1, firstly we extract a diagonal matrix $\mat{D}_1$ and repeat the same process.
Then we get a block-triangular matrix.

\subsubsection{Extraction of Non-Singular Diagonal Blocks \label{sssec:nsdb}}
Consider the case that matrix $\mat{W}$ does not contain columns of weight 1.
If we choose appropriate $m_1$, singly bordered block-diagonalization (Sect.~\ref{sec:2-c}) to $\mat{W}$ gives
\begin{equation*}
  \mat{W}' = \mat{P}\mat{W}\mat{Q}
  =
  \begin{pmatrix}
    \mat{B}_1 & \mat{O}   & \mat{Z}_1 \\
    \mat{O}   & \mat{B}_2 & \mat{Z}_2 \\    
  \end{pmatrix},
\end{equation*}
where $\mat{B}_1$ is $m_1 \times n_1$ full-rank horizontal submatrix.
Since $\rank(\mat{B}_1) = m_1$, we extract a non-singular $m_1\times m_1$ matrix $\mat{F}_1$ by rearranging the rows of $\mat{B}_1$ as follows:
\begin{equation*}
  \mat{W}'\mat{Q}'
  =
  \left(
  \begin{array}{c|ccc}
    \mat{F}_1 & \mat{R}_1 & \mat{O}   & \mat{Z}_1 \\ \hline
    \mat{O}   & \mat{O}   & \mat{B}_2 & \mat{Z}_2 \\    
  \end{array}
  \right)
  =
  \begin{pmatrix}
    \mat{F}_1 & \mat{M}_1 \\
    \mat{O}   & \mat{W}_1 \\
  \end{pmatrix}.
\end{equation*}
Summarizing above, singly bordered block-diagonalization allows us to extract a non-singular matrix for binary and non-binary codes.

In the case of non-binary codes, we can also extract a non-singular matrix by cycle detection in the associate graph.
For a given matrix $\mat{W}$, construct a submatrix $\mat{\widetilde{C}}$ which consists of all the columns of weight 2 in $\mat{W}$.
We detect one of the smallest cycles (e.g., see \cite{itai1978finding}) in the associated graph for $\mat{\widetilde{C}}$.
Since cycles in an associated graph correspond to the cycle matrices, we get the following matrix by moving the rows and columns corresponding to the smallest cycle in the associated graph:
\begin{equation*}
  \begin{pmatrix}
    \mat{C}_1 & \mat{M}_1 \\
    \mat{O}   & \mat{W}_1 \\
  \end{pmatrix},
\end{equation*}
where $\mat{C}_1$ is a cycle matrix.
Hence, we can extract a cycle, i.e., non-singular, matrix by cycle detection in the associate graph.
Note that this method cannot extract a non-singular matrix for the binary code by the reason in Remark \ref{re:cycle}.

\subsection{Preprocessing Stage of Proposed EA}\label{sec:3-2}
Now, we explain the details of the preprocessing stage of the proposed EA.
This stage is realized by an iterative algorithm.
Let $t$ stand for the round of this iterative algorithm.
Let $\mat{H}^{(t)}$ be the parity-check matrix at the $t$-th round.
This stage performs row and column permutations to $\mat{H}^{(t)}$ at each round.
This stage determines $\mat{H}'$ given in Eq.~\eqref{eq:H'} from front columns and rear columns at each round.
We denote the number of determined front columns and rear columns at round $t$, by $f_t$ and $g_t$, respectively. 

As an example, Figure \ref{fig1} depicts $\mat{H}^{(3)}$.
The pink blocks of Fig.~\ref{fig1} express determined submatrices.
We refer to the blue block of Fig.~\ref{fig1}, i.e., $\mat{W}^{(t)}$, as \textit{working space}.

\begin{figure}[t]
\centering
\includegraphics[width=78mm]{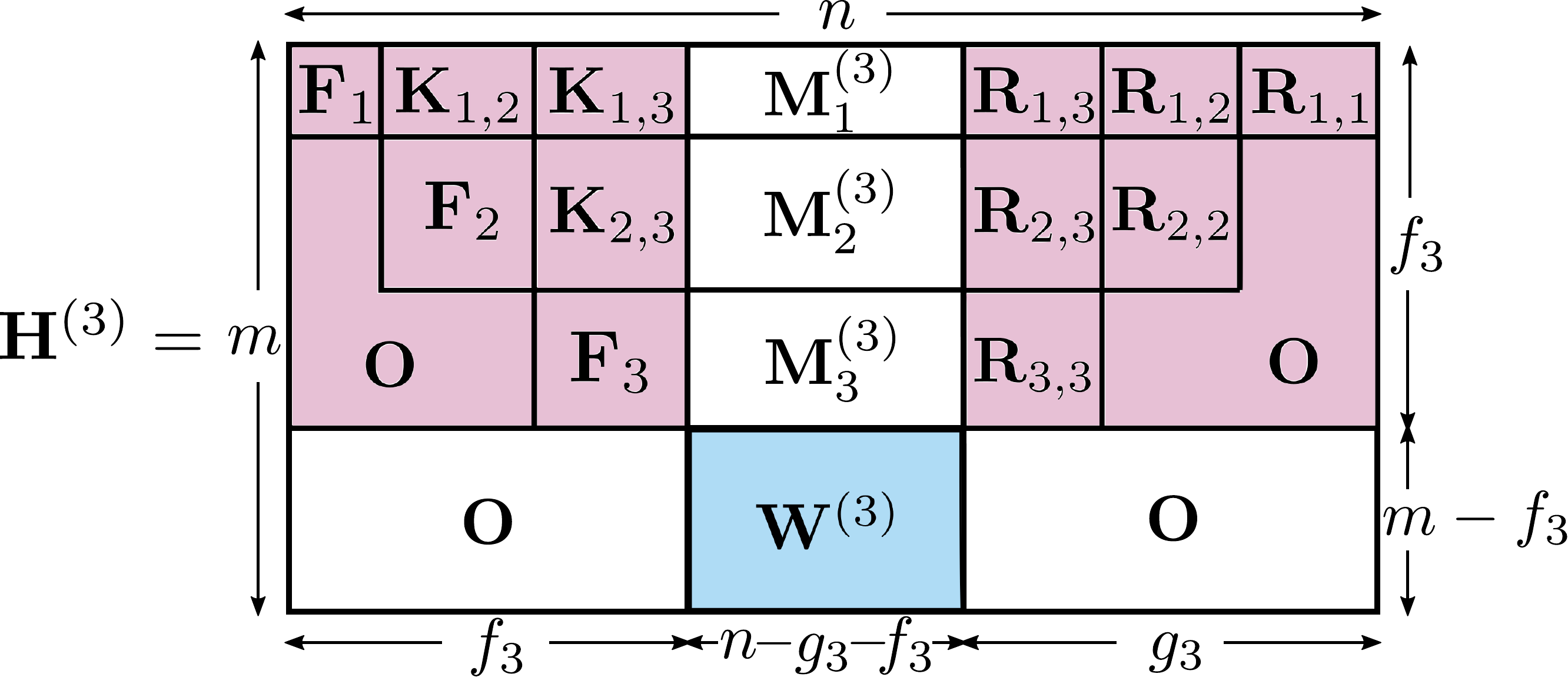} 
\caption{A matrix $\mat{H}^{(3)}$}
\label{fig1}
\end{figure}

\begin{figure}[!t]  
\centering
\includegraphics[width=\linewidth]{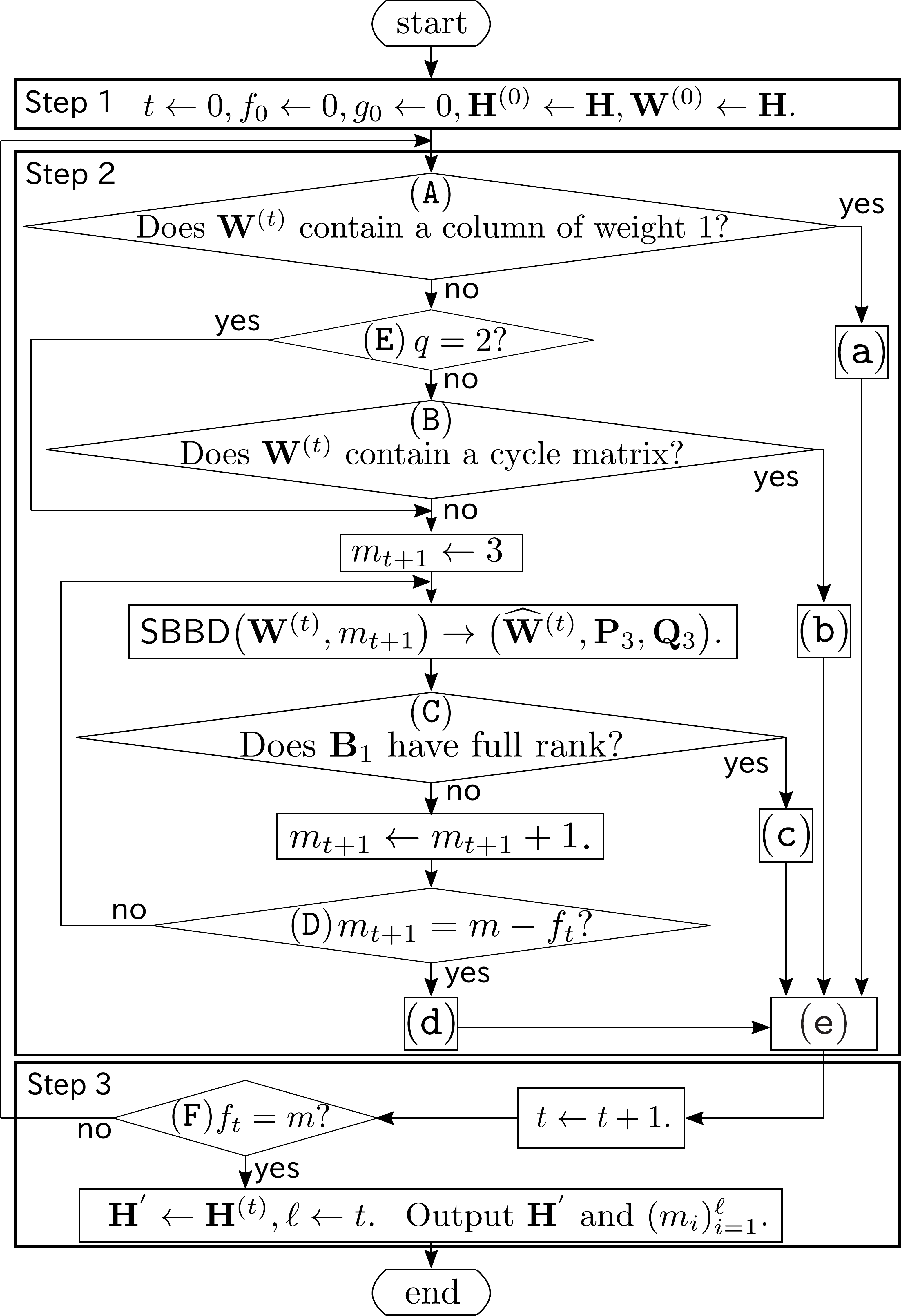}
\caption{Flowchart of preprocessing stage of proposed EA}
\label{fig2}
\end{figure}

Figure \ref{fig2} shows the flowchart of the preprocessing stage.
This stage is divided into three steps: initialization (Step 1), iteration (Step 2), and end process (Step 3).
In Step 2, perform row and column permutations to $\mat{H}^{(t)}$ at round $t$.
Roughly speaking, Step 2 forms the diagonal block $\mat{F}_{t+1}$ at round $t$.
Condition (\texttt{A}) gives a decision whether we can make $\mat{F}_{t+1}$ a diagonal matrix.
Condition (\texttt{E}) and (\texttt{B}) decide whether we can make $\mat{F}_{t+1}$ a cycle matrix.
Unless $\mat{F}_{t+1}$ becomes a diagonal or a cycle matrix, we form $\mat{F}_{t+1}$ into a small ATM.
To make a small ATM, we adjust $m_{t+1}$, which decides the size of $\mat{F}_{t+1}$.
If $\mat{F}_{t+1}$ cannot become a small ATM, i.e., Condition (\texttt{D}) is satisfied, we make a large ATM in Sub-step (\texttt{d}).
In Step 3, this stage outputs the result if this stage satisfies Condition (\texttt{F}), i.e., $f_t = m$.
Otherwise, return to Step 2.
The details of Sub-steps (\texttt{a})-(\texttt{e}) are as follows.

\paragraph{Sub-step (\texttt{a})}
This sub-step forms a diagonal matrix at the upper-left of $\mat{W}^{(t)}$.
We write $d_t$ for the number of the columns of weight 1 in $\mat{W}^{(t)}$.
This stage performs the column permutation $\mat{Q}_1$ into $\mat{W}^{(t)}$ for moving $d_t$ columns of weight 1 to the front of $\mat{W}^{(t)}$.
This stage applies the row permutation $\mat{P}_1$ for moving these non-zero entries to the upper-left of $\mat{W}^{(t)}$.
As a result, this stage gets
\begin{equation}
  \label{eq:a1}
  \mat{P}_1\mat{W}^{(t)}\mat{Q}_1 = 	
  \begin{pmatrix}
    \mat{\widetilde D}_{t+1} & \mat{M}_{t+1}^{(t+1)}  \\
    \mat{O}    &  \mat{W}^{(t+1)}  \\
  \end{pmatrix},
\end{equation}
where all the columns of $\mat{\widetilde D}_{t+1} \in \mathbb{F}^{m_{t+1} \times d_{t}}_q$ are weight 1.

Next, to extract an $m_{t+1} \times m_{t+1}$ diagonal matrix from $\mat{\widetilde D}_{t+1}$, move $m_{t+1}$ columns forming a diagonal matrix to the front of $\mat{W}^{(t)}$ and move the residual ($d_{t} - m_{t+1}$) columns to the rear of $\mat{W}^{(t)}$ by $\mat{Q}_2$, i.e.,
\begin{equation}
  \label{eq:a2}
  \mat{P}_1\mat{W}^{(t)}\mat{Q}_1\mat{Q}_2 = 	
  \begin{pmatrix}
    \mat{D}_{t+1} & \mat{M}_{t+1}^{(t+1)} & \mat{R}_{t+1,t+1} \\
    \mat{O}    &  \mat{W}^{(t+1)} & \mat{O} \\
  \end{pmatrix},
\end{equation}
where $\mat{D}_{t+1}$ is an $m_{t+1} \times m_{t+1}$ diagonal matrix and $\mat{R}_{t+1,t+1}$ is an $m_{t+1} \times (d_t - m_{t+1})$ matrix.
Then, set $f_{t+1} \gets f_t + m_{t+1}$, $g_{t+1} \gets g_t + (d_t - m_{t+1})$, $\mat{P}^{(t)}=\mat{P}_1$, and $\mat{Q}^{(t)}=\mat{Q}_1\mat{Q}_2$.

\paragraph{Sub-step (\texttt{b})}
This sub-step extracts one of the smallest cycle matrices at the upper-left of $\mat{W}^{(t)}$.
First, this stage detects one of the smallest cycle matrices in $\mat{W}^{(t)}$.
Let $c_t$ be the number of the columns of weight 2 in $\mat{W}^{(t)}$.
Denote the $(m - f_t) \times c_t$ submatrix which consists of all the columns of weight 2, by $\mat{\widetilde C}_{t+1}$.
This stage detects the smallest cycle (e.g., see \cite{itai1978finding}) in the associated graph for $\mat{\widetilde C}_{t+1}$.

Next, this stage forms the cycle matrix at the upper-left of $\mat{W}^{(t)}$ by moving the rows and columns corresponding to the smallest cycle in the associated graph:
\begin{equation}
  \label{eq:b1}
  \mat{P}^{(t)}\mat{W}^{(t)}\mat{Q}^{(t)} = 	
  \begin{pmatrix}
    \mat{C}_{t+1} & \mat{M}_{t+1}^{(t+1)} \\
    \mat{O}    &  \mat{W}^{(t+1)} \\
  \end{pmatrix},
\end{equation}
where $\mat{C}_{t+1}$ is the $m_{t+1} \times m_{t+1}$ cycle matrix.
Set $f_{t+1} \gets f_t + m_{t+1}$ and $g_{t+1} \gets g_t$.

\paragraph{Sub-step (\texttt{c})}
The purpose of this sub-step is to form a small ATM at the upper-left of $\mat{W}^{(t)}$.
Let $n_{t+1}$ be the number of columns of $\mat{B}_1$.
Execute $\mathsf{ATM}(\mat{B}_1) \to (\mat{A}_{t+1},\mat{P}_4,\mat{Q}_4)$ and get $\mat{\widetilde W}^{(t)}$ as
\begin{align}
  \label{eq:c1}
  \mat{\widetilde W}^{(t)} 
  &=
  \begin{pmatrix} \mat{P}_4 & \mat{O} \\ \mat{O} & \mat{I}_{m - f_t - m_{t+1}} \end{pmatrix}
  \mat{\widehat W}^{(t)}
  \begin{pmatrix} \mat{Q}_4 & \mat{O} \\ \mat{O} & \mat{I}_{n - f_t - g_t - n_{t+1}} \end{pmatrix}
  \notag \\&=
  \left(
  \begin{array}{cccc}
    \mat{A}_{t+1} & \mat{R}_{t+1,t+1} & \mat{O} & \ \ \widetilde{\mat{Z}}_1
     \\
    \mat{O}    & \mat{O} & \mat{B}_2 & \ \ \mat{Z}_2  \\
  \end{array}
  \right),
\end{align}
where $\mat{A}_{t+1}$ is an $m_{t+1} \times m_{t+1}$ ATM, $\mat{I}_k$ represents the $k\times k$ identity matrix, and $\widetilde{\mat{Z}}_1 := \mat{P}_4\mat{Z}_1$.
This stage moves $\mat{R}_{t+1,t+1}$ to the rear of the matrix by a suitable column permutation $\mat{Q}_5$:
\begin{align}
  \mat{\widetilde W}^{(t)}\mat{Q}_5 &=
  \begin{pmatrix}
    \mat{A}_{t+1} & \mat{O} & \widetilde{\mat{Z}}_1 & \mat{R}_{t+1,t+1}\\
    \mat{O} & \mat{B}_2 & \mat{Z}_2 & \mat{O}
  \end{pmatrix}\nonumber
  \\&=
  \begin{pmatrix}
    \mat{A}_{t+1} & \mat{M}_{t+1}^{(t+1)} & \mat{R}_{t+1,t+1} \\
    \mat{O}    &  \mat{W}^{(t+1)} & \mat{O} \\
  \end{pmatrix}\label{eq:c2}.
\end{align}
Set $f_{t+1} \gets f_t + m_{t+1}$, $g_{t+1} \gets g_t + (n_{t+1} - m_{t+1})$ and
\begin{align*}
&\mat{P}^{(t)} :=
\begin{pmatrix}
  \mat{P}_4 & \mat{O}\\
  \mat{O} & \mat{I}_{m-f_{t+1}}
\end{pmatrix}
\mat{P}_3,
\\
&\mat{Q}^{(t)}:=\mat{Q}_3
\begin{pmatrix}
  \mat{Q}_4 & \mat{O}\\
  \mat{O} & \mat{I}_{n-f_t-g_t-n_{t+1}}
\end{pmatrix}
\mat{Q}_5.
\end{align*}

\paragraph{Sub-step (\texttt{d})}
Execute $\mathsf{ATM}\bigl(\mat{W}^{(t)}\bigr) \to \bigl(\mat{A}_{t+1},\mat{P}^{(t)},\mat{Q}^{(t)}\bigr)$
and get
\begin{equation*}
  \mat{P}^{(t)}\mat{W}^{(t)}\mat{Q}^{(t)} = 	
  \begin{pmatrix}
    \mat{A}_{t+1} & \mat{R}_{t+1,t+1} \\
  \end{pmatrix},
\end{equation*}
where $\mat{A}_{t+1}$ is the $m_{t+1} \times m_{t+1}$ ATM.
Set $f_{t+1} \gets f_t + m_{t+1}$ and $g_{t+1} \gets g_t$.

\paragraph{Sub-step (\texttt{e})}
Perform
\begin{equation*}
  \mat{H}^{(t+1)} \gets    
  \begin{pmatrix}
    \mat{I}_{f_{t}} & \mat{O} \\
    \mat{O}    & \mat{P}^{(t)} \\
  \end{pmatrix}
  \mat{H}^{(t)}
  \begin{pmatrix}
    \mat{I}_{f_{t}} & \mat{O} & \mat{O} \\
    \mat{O}    & \mat{Q}^{(t)} & \mat{O} \\
    \mat{O}    & \mat{O} & \mat{I}_{g_{t}} \\
  \end{pmatrix}.
\end{equation*}
By the permutation $\mat{Q}^{(t)}$, the matrix $\mat{M}_{i}^{(t)}$ for $1\le i \le t$ is transformed as 
$\mat{M}_{i}^{(t)}\mat{Q}^{(t)} = \Bigl(\mat{K}_{i,t+1} \,\ \mat{M}_{i}^{(t+1)} \,\ \mat{R}_{i,t+1}\Bigr)$,
where $\mat{K}_{i,t+1}$, $\mat{M}_{i}^{(t+1)}$ and $\mat{R}_{i,t+1}$ are of size $m_i\times m_{t+1}$, $m_i \times (n-f_{t+1}-g_{t+1})$ and $m_{i} \times (g_{t+1}-g_{t})$, respectively.

\begin{re.}\label{re:subm}
  In Sub-steps (\texttt{c}) and (\texttt{d}), we can use any preprocessing stage of EAs for binary and non-binary irregular LDPC codes (e.g., \cite{kaji2006encoding}).
  If we use the K-EA \cite{kaji2006encoding}, for ATMs with small gaps, we can reduce the complexity of proposed EA as shown in Sect.\ \ref{sec:4}.
\end{re.}

\begin{ex.}
\begin{figure}[!t]
  \centering
  \includegraphics[width=\linewidth]{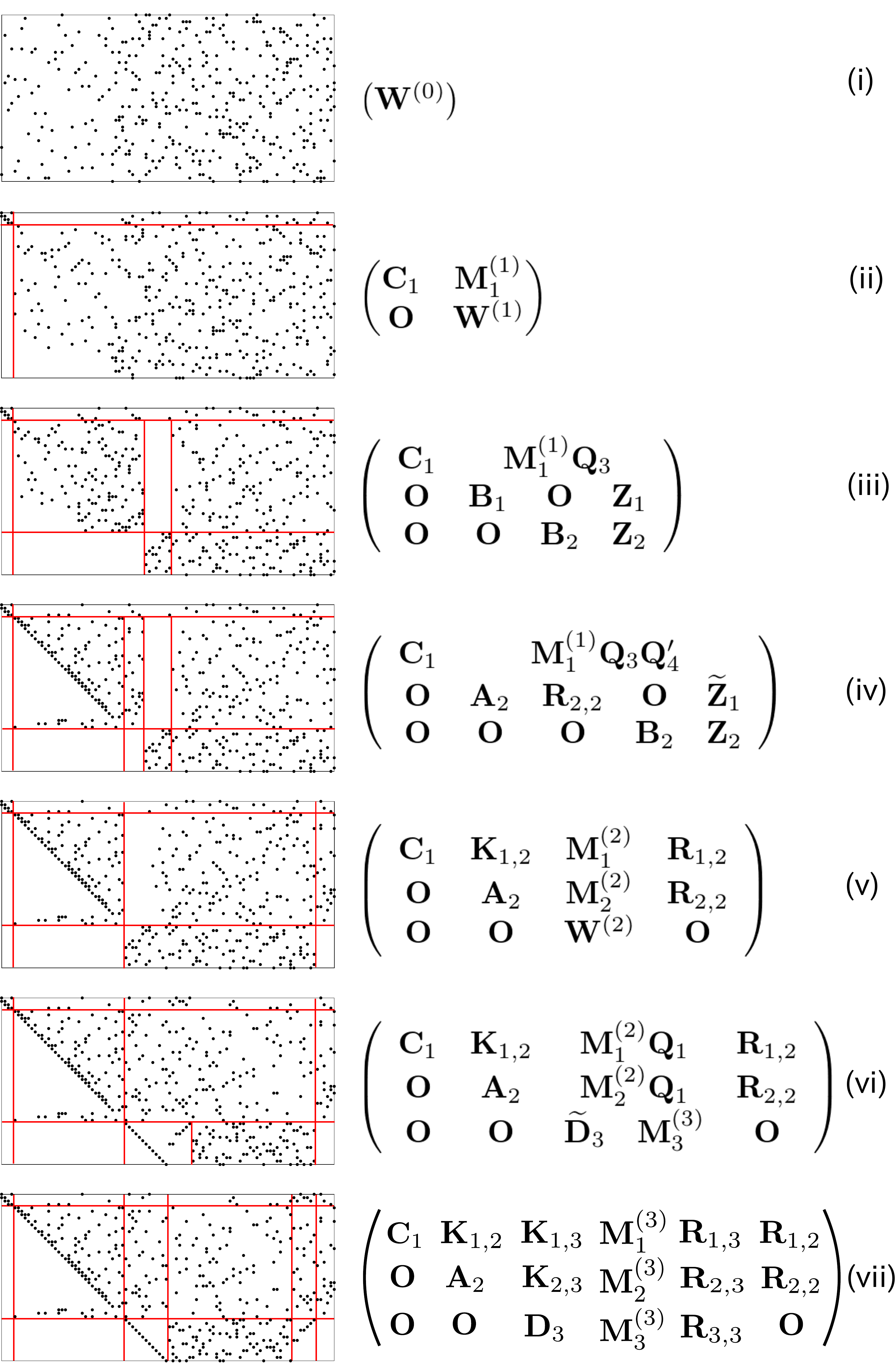}
  \caption{Demonstration of preprocessing stage}
  \label{fig:ex1}
\end{figure}

This example demonstrates this preprocessing stage.
The left figure of Fig.~\ref{fig:ex1}-(i) represents the input matrix.
In the left figures, the black dots stand non-zero elements and the white parts express zero elements.
Sub-blocks separated by red lines in the left figures correspond to submatrices in the right matrices.

We input a parity-check matrix $\mat{H} \in \mathbb{F}^{50 \times 100}_8$ as Fig.~\ref{fig:ex1}-(i).
This matrix contains $33$ columns of weight $2$, $67$ columns of weight $5$, $49$ rows of weight $8$, and a row of weight $9$.
In Step 1, set $t \gets 0$, $f_0 \gets 0$, $g_0 \gets 0$, $\mat{H}^{(0)} \gets \mat{H}$, and $\mat{W}^{(0)} \gets \mat{H}$.
Note that Condition (\texttt{E}) of Step 2 is always unsatisfied since $q=8$.

{\bf ($\mathbf{t=0}$)~~}
In Step 2, Condition (\texttt{A}) is not satisfied since $\mat{W}^{(0)}= \mat{H}$ does not contain a column of weight $1$.
Hence, go to the decision of Condition (\texttt{B}).
Condition (\texttt{B}) is satisfied because $\mat{W}^{(0)}$ contains a cycle matrix.
Hence, execute Sub-step (\texttt{b}).

In the execution of Sub-step (\texttt{b}), search one of the smallest cycle matrices in $\mat{W}^{(0)}$.
The smallest cycle matrix is of size $4 \times 4$.
By Eq.~\eqref{eq:b1}, transform Fig.~\ref{fig:ex1}-(i) into Fig.~\ref{fig:ex1}-(ii).
Set $f_1 \gets 4$ and $g_1 \gets 0$.
Go to Step 3.

In Step 3, increase $t \gets 1$.
Condition (\texttt{F}) is not satisfied since $f_1=4 \neq 50$.
Hence, return to Step 2.

{\bf ($\mathbf{t=1}$)~~}
In Step 2, Conditions (\texttt{A}) and (\texttt{B}) are not satisfied because $\mat{W}^{(1)}$ does not contain a column of weight $1$ and a cycle matrix.
Next, this stage increases $m_2$ until submatrix $\mat{B}_1$, obtained by $\mathsf{SBBD}\bigl(\mat{W}^{(1)},m_2\bigr) \to \mat{\widehat W}^{(1)}$, becomes full rank.
In this example, we get full rank submatrix $\mat{B}_1$ by executing $\mathsf{SBBD}\bigl(\mat{W}^{(1)},m_2=37\bigr) \to \mat{\widehat W}^{(1)}$.
By this execution, the parity-check matrix is transformed as Fig.~\ref{fig:ex1}-(iii).
Hence, go to Sub-step (\texttt{c}).

Transform Fig.~\ref{fig:ex1}-(iii) into Fig.~\ref{fig:ex1}-(iv) by Eq.~\eqref{eq:c1}, where 
$
\mat{Q}'_4 =
\begin{pmatrix}
  \mat{Q}_4 & \mat{O}\\
  \mat{O} & \mat{I}_{53}
\end{pmatrix}
$.
Then, Fig.~\ref{fig:ex1}-(iv) is transformed into Fig.~\ref{fig:ex1}-(v) by Eq.~\eqref{eq:c2}.
Set $f_2 \gets 41$ and $g_2 \gets 6$.
Go to Step 3.

In Step 3, increase $t \gets 2$.
Condition (\texttt{F}) is not satisfied since $f_2 = 41 \neq 50$.
Hence, return to Step 2.

{\bf ($\mathbf{t=2}$)}
In Step 2, because Condition (\texttt{A}) is satisfied, Sub-step (\texttt{a}) is executed.
Transforms Fig.~\ref{fig:ex1}-(v) into Fig.~\ref{fig:ex1}-(vi) by Eq.~\eqref{eq:a1}.
Moreover, transform Fig.~\ref{fig:ex1}-(vi) into Fig.~\ref{fig:ex1}-(vii) by Eq.~\eqref{eq:a2}.
Set $f_3 \gets 50$ and $g_3 \gets 13$.
Go to Step 3.

In Step 3, increase $t \gets 3$.
Condition (\texttt{F}) is satisfied since $f_3=50=m$.
Output $m_1=4, m_2=37, m_3=9$ and the matrix given in Fig.~\ref{fig:ex1}-(vii).
Here, the parity part $\mat{H}'_{P}$ and message part $\mat{H}'_I$ of the output matrix are
\begin{equation*}
  \begin{small}
  \mat{H}'_{P}
  =
  \begin{pmatrix}
    \mat{C}_1 & \mat{K}_{1,2} & \mat{K}_{1,3} \\
    \mat{O}   & \mat{A}_{2} & \mat{K}_{2,3} \\
    \mat{O}   & \mat{O}    & \mat{D}_{3} \\
  \end{pmatrix},
  \quad 
  \mat{H}'_{I}
  =
  \begin{pmatrix}
    \mat{M}_1^{(3)} & \mat{R}_{1,3} & \mat{K}_{1,2} \\
    \mat{M}_2^{(3)} & \mat{R}_{2,3} & \mat{K}_{2,2} \\
    \mat{M}_3^{(3)} & \mat{R}_{3,3} & \mat{O} \\
  \end{pmatrix}.
  \end{small}
\end{equation*}
\end{ex.}

\subsection{Encoding Stage of Proposed EA}\label{sec:3-3}
Recall that the preprocessing stage outputs the number of rows of submatrices $m_1, m_2, \dots, m_{\ell-1}, m_{\ell}$.
According to these values, we split the parity part $\vec{p}$ of the codeword into $\ell$ parts $\vec{p}_i \in \mathbb{F}^{m_i}_q$ ($i = 1, 2, \dots, \ell$).
Then, the codeword is expressed as $( \vec{p}_1 \, \vec{p}_2 \, ... \, \, \vec{p}_{\ell} \, \vec{u})$.
Combining Eq.~\eqref{eq:H'} and $\mat{H}'\vec{x}^T=\vec{0}^T$, we get the following system of linear equations:
\begin{align}
  \mat{F}_{1}\vec{p}^T_1 
  &= \vec{b}_{1}^T
  \hspace{2.2mm}:=
  - \mat{H}'_{I,1}\vec{u}^T
  -{\textstyle\sum_{j=2}^{\ell}\bigl(\mat{K}_{1,j}\vec{p}^T_{j}\bigr)}, 
  \label{eq:ren1} \\
  \mat{F}_{2}\vec{p}^T_2
  &= \vec{b}_{2}^{T}
  \hspace{2.2mm}:= - \mat{H}'_{I,2}\vec{u}^T
  -{\textstyle\sum_{j=3}^{\ell}\bigl(\mat{K}_{2,j}\vec{p}^T_{j}\bigr)}, 
  \label{eq:ren2} \\
  & \vdots
   \notag \\
  \mat{F}_{\ell-1}\vec{p}^T_{\ell-1}
  &= \vec{b}^T_{\ell-1}
  := - \mat{H}'_{I,\ell-1}\vec{u}^T -\mat{K}_{\ell-1, \ell}\vec{p}^T_{\ell}, 
  \label{eq:renl-1}\\
  \mat{F}_{\ell}\vec{p}^T_{\ell}
  &= \vec{b}^T_{\ell}
  \hspace{2.4mm}:= - \mat{H}'_{I,\ell}\vec{u}^T . 
  \label{eq:renl}
\end{align}
Firstly, we solve Eq.~\eqref{eq:renl} by efficient algorithm to solve $\mat{F}_{\ell} \vec{p}^T_{\ell}= \vec{b}^T_{\ell}$.
Next, substituting $\vec{p}_{\ell}$ into Eq.~\eqref{eq:renl-1}, we solve Eq.~\eqref{eq:renl-1}, i.e., $\mat{F}_{\ell-1} \vec{p}^T_{\ell-1}= \vec{b}^T_{\ell-1}$.
Similarly, we solve $\mat{F}_{i} \vec{p}_i^T = \vec{b}^T_i$ for $\ell-2, \ell-3, \dots, 1$.
Finally, we get the parity part $\vec{p} = (\vec{p}_1, \vec{p}_2, \dots \vec{p}_{\ell})$.

\subsection{Complexity of Proposed EA}\label{sec:3-4}
By summarizing Eqs.~\eqref{eq:def-fm}, \eqref{eq:def-fa}, and \eqref{eq:cyc_comp},
the number of multiplications $\mu_i$ and additions $\alpha_i$ to solve $\mat{F}_{i} \vec{p}^T_i = \vec{b}^T_i$ are
\begin{align*}
  &(\mu_i, \alpha_i) =
  \begin{cases}
    (m_i, 0), & (\mat{F}_{i}: \text{diag}),\\
    (3m_i - 1, 2m_i -2), & (\mat{F}_{i}: \text{cycle}),\\
    (f_m(\mat{F}_i), f_a(\mat{F}_i)),  & (\mat{F}_{i}: \text{ATM}).
  \end{cases}
\end{align*}
This notation gives the total number of multiplications $\mu'$ and additions $\alpha'$ for the encoding stage of the proposed EA:
\begin{align*}
  \mu' &= {\sum_{i=1}^{\ell} \left\{ \mu_i + \wt\Bigl(\bigl(\mat{K}_{i,i+1} \,\, \mat{K}_{i,i+2} \,\, \cdots \,\, \mat{K}_{i,\ell} \,\, \mat{H}'_{I,i}\bigr)\Bigl) \right\}}, \\
  \alpha' &= {\sum_{i=1}^{\ell} \left\{ \alpha_i + \mathcal{S}\Bigl(\bigl(\mat{K}_{i,i+1} \,\, \mat{K}_{i,i+2} \,\, \cdots \,\, \mat{K}_{i,\ell} \,\, \mat{H}'_{I,i}\bigr)\Bigr) \right\}}.
\end{align*}

\subsection{Property of Proposed EA}\label{sec:3-5}
In this section, we prove that the proposed EA is a generalization of the HZ-EA.
When the associated graph is a connected graph, we say the cycle code is \textit{proper}.
The following theorem shows that if the input of preprocessing stage is proper cycle code, the output matrix satisfies the same properties of the HZ-EA.
\begin{theorem}
  If the input matrix $\mat{H}$ is a non-binary parity-check matrix for a proper cycle code, the preprocessing stage of the proposed EA outputs $\mat{H}'$ satisfying
  (i) $\mat{F}_{1} = \mat{C}_{1}$,
  (ii) $\mat{F}_{i} = \mat{D}_{i}$ ($2\le i \le \ell$),
  (iii) $\mat{K}_{i-1,i} \neq \mat{O}$ ($2 \le i \le \ell$), and
  (iv) $\mat{K}_{j,i} = \mat{O}$ ($2 \le i \le \ell$, $1 \le j \le i-2$).
\end{theorem}
\begin{proof}
  First, we will show $\mat{F}_{1} = \mat{C}_{1}$.
  Note that $\mat{W}^{(0)} = \mat{H}$.
  Hence, in Step 2, Condition (\texttt{A}) is not satisfied since $\mat{H}$ is the parity check matrix of the cycle code.
  Since $\mat{H}$ is non-binary, Condition (\texttt{E}) is not satisfied.
  Consider the associate graph of $\mat{W}^{(0)}$.
  Note that the associate graph consists of $m$ nodes and $n$ edges.
  Since $n > m$, the associate graph includes at least one cycle.
  This leads that Condition (\texttt{B}) is satisfied.
  Hence, Sub-step (\texttt{b}) is executed.
  Then, $\mat{H}^{(1)}$ is transformed as follows:
  \begin{equation*}
    \mat{H}^{(1)}=
    \mat{P}^{(0)}\mat{H}^{(0)}\mat{Q}^{(0)} = 	
    \begin{pmatrix}
      \mat{C}_{1} & \mat{M}_{1}^{(1)}\\
      \mat{O}    &  \mat{W}^{(1)}\\
    \end{pmatrix}.
  \end{equation*}
  This equation leads $\mat{F}_{1} = \mat{C}_{1}$.

  Hypothesize $\mat{W}^{(1)}$ does not contain a column of weight 1.
  Then, each column has weight 0 or 2.
  If the $i$-th column of $\mat{W}^{(1)}$ has weight 0 (resp.\ 2), the $i$-th column of $\mat{M}_1^{(1)}$ has weight 2 (resp.\ 0).
  This leads that the associate graph of $\mat{H}$ is disconnected.
  This contradicts the input parity-check matrix is for a proper cycle code.
  Hence, $\mat{W}^{(1)}$ contains at least one column of weight 1.

  Since $\mat{W}^{(1)}$ contains a column of weight 1, Sub-step (\texttt{a}) and (\texttt{e}) is executed.
  The resulting matrix $\mat{H}^{(2)}$ is
  \begin{equation*}
    \mat{H}^{(2)}=
    \begin{pmatrix}
      \mat{C}_{1} & \mat{K}_{1,2} & \mat{M}_{1}^{(2)} & \mat{R}_{1,2} \\
      \mat{O}    & \mat{D}_2 & \mat{M}_2^{(2)} & \mat{R}_{2,2} \\
      \mat{O}    & \mat{O} & \mat{W}^{(2)} & \mat{O}\\      
    \end{pmatrix}.
  \end{equation*}
  From the procedure of Sub-step (\texttt{a}), all the columns of $\mat{K}_{1,2}$ have weight 1.
  Moreover, the columns of $\mat{M}_1^{(2)}$ have weight 0 or 2.
  Hence, $\mat{K}_{1,2} \neq 0$.
  
  Similar to the above, $\mat{W}^{(2)}$ contains at least one column of weight 1.
  Hence, Sub-step (\texttt{a}) is executed, and the resulting matrix $\mat{H}^{(3)}$ is
  \begin{equation*}
    \mat{H}^{(3)}=
    \begin{pmatrix}
      \mat{C}_{1} & \mat{K}_{1,2} & \mat{K}_{1,3} & \mat{M}_{1}^{(3)} & \mat{R}_{1,3} & \mat{R}_{1,2} \\
      \mat{O}    & \mat{D}_{2}   & \mat{K}_{2,3} & \mat{M}_{2}^{(3)} & \mat{R}_{2,3} & \mat{R}_{2,2} \\      
      \mat{O}    & \mat{O}      & \mat{D}_3    & \mat{M}_{3}^{(3)} & \mat{R}_{3,3} & \mat{O} \\
      \mat{O}    & \mat{O}      & \mat{O}      & \mat{W}^{(3)}    & \mat{O}      & \mat{O} \\
    \end{pmatrix}.
  \end{equation*}
  From the procedure of Sub-step (\texttt{a}), submatrix $\begin{pmatrix} \mat{K}_{1,3} & \mat{M}_1^{(3)} & \mat{R}_{1,3}\end{pmatrix}$ is obtained from a column permutation to $\mat{M}_{1}^{(2)}$.
  Hence, all the columns of $\mat{K}_{1,3}, \mat{M}_1^{(3)}, \mat{P}_{1,3}$ have weight 0 or 2.
  Since all the columns of $\mat{D}_3$ have weight 1, all the columns of $\mat{K}_{1,3}$ have weight 0 and all the columns of $\mat{K}_{2,3}$ have weight 1.
  Hence, $\mat{K}_{1,3}=\mat{O}$ and $\mat{K}_{2,3} \neq \mat{O}$.
  Moreover, from the procedure of Sub-step (\texttt{a}), the columns of $\mat{M}_2^{(3)}$ have weight 0 or 2.
  
  In a similar matter, we can show that for $i=4,5,\dots, \ell$: (ii) $\mat{F}_i = \mat{D}_i$, (iii) $\mat{K}_{i-1,i} \neq \mat{O}$, and (iv) $\mat{K}_{j,i} = \mat{O}$ ($j =1,2,\dots, i-2$).
\end{proof}

\section{Numerical Experiments}\label{sec:4}
This section compares the proposed EA, the RU-EA \cite{richardson2001efficient}, the K-EA \cite{kaji2006encoding} for the number of additions $\alpha$ and multiplications $\mu$ by numerical experiments.

\subsection{Comparison with Existing EAs}
\begin{table*}[t]
  \centering
  \caption{Average encoding complexity for $\mathcal{E}_8$ and $\mathcal{E}_2$ \label{tab:ave_comp}}
  \begin{tabular}{|c||c|c|c||c|c|c|}     \hline
    & \multicolumn{3}{c||}{Complexity for $\mathcal{E}_8$}
    & \multicolumn{3}{c|}{Complexity for $\mathcal{E}_2$}  \\ \hline
    $n$ & $\alpha_{\rm RU}/\mu_{\rm RU}$ (RU) & $\alpha_{\rm K}/\mu_{\rm K}$ (Kaji) & $\alpha'/\mu'$ (Proposed)
    & $\alpha_{\rm RU}$ (RU) & $\alpha_{\rm K}$ (Kaji) & $\alpha'$ (Proposed)
    \\ \hline
    $1000$    & $2614.82/3613.71$ & $1613.17/3110.06$     & ${\bf 1576.55}/{\bf 2577.55}$   
      & $4054.16$   & ${\bf 3195.16}$ & $3254.76$   \\ \hline
    $2000$    & $5234.19/7229.83$ & $3207.56/6201.20$     & ${\bf 3151.48}/{\bf 5152.48}$   
      & $8217.58$   & $7216.17$ & ${\bf 6670.03}$   \\ \hline
    $3000$    & $7853.93/10846.39$ & $4797.21/9287.67$     & ${\bf 4726.37}/{\bf 7727.37}$    
      & $12413.73$  & $12157.26$ & ${\bf 10190.95}$ \\ \hline
    $4000$    & $10472.86/14463.01$ & $6382.11/12370.26$     & ${\bf 6301.37}/{\bf 10302.37}$     
      & $16642.75$  & $17592.54$ & ${\bf 13795.31}$ \\ \hline
    $5000$    & $13095.01/18081.69$ & $7968.67/15453.35$     & ${\bf 7876.35}/{\bf 12877.35}$   
      & $20892.28$  & $24118.86$ & ${\bf 17473.06}$ \\ \hline
  \end{tabular}
\end{table*}

\begin{figure*}[t]
  \centering
  \subfigure[Additions for $\mathcal{E}_8$]{%
    \includegraphics[clip,width=.32\linewidth]{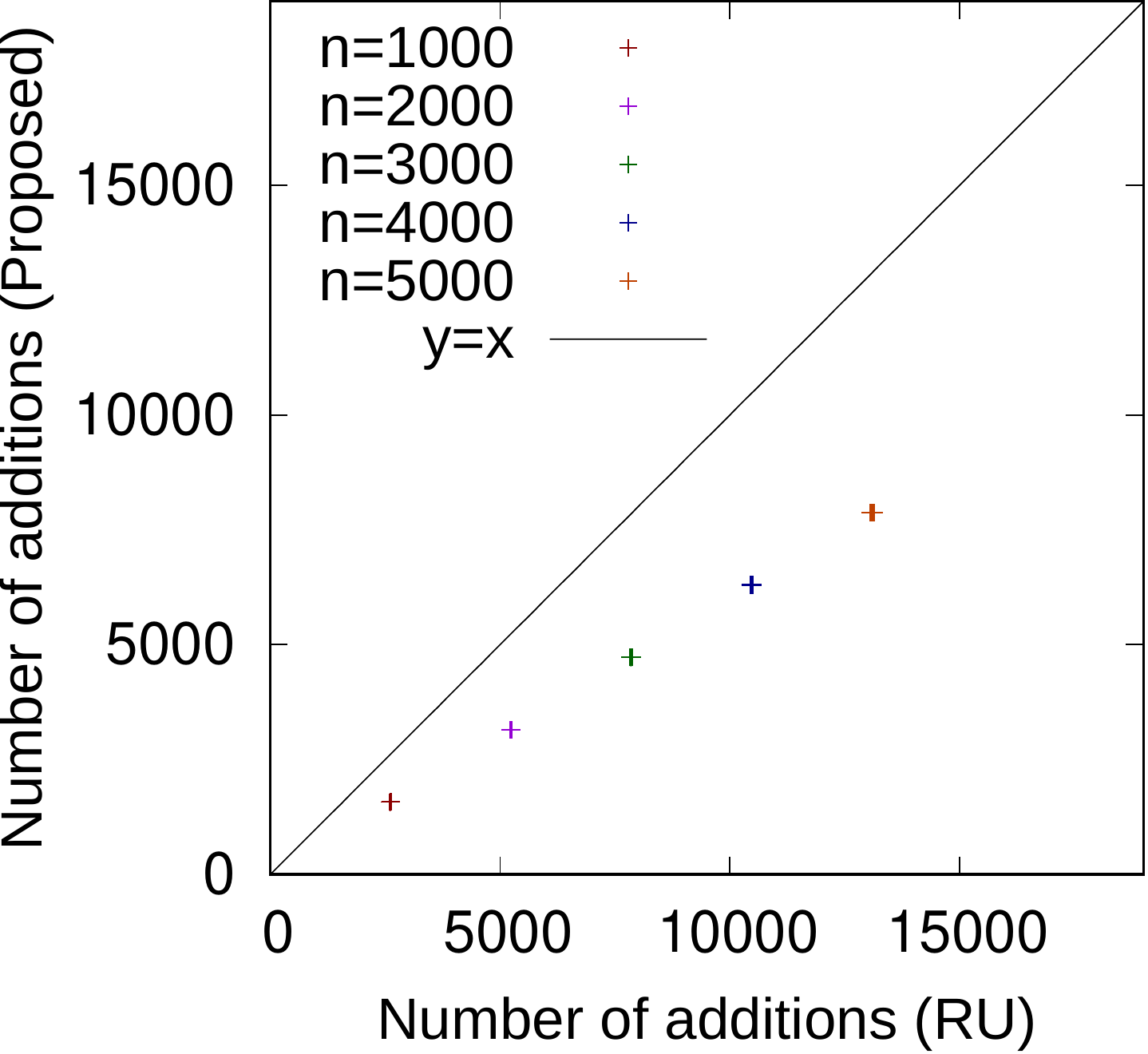} \label{fig:num1ar}
  }
  \subfigure[Multiplications for $\mathcal{E}_8$]{%
    \includegraphics[clip,width=.32\linewidth]{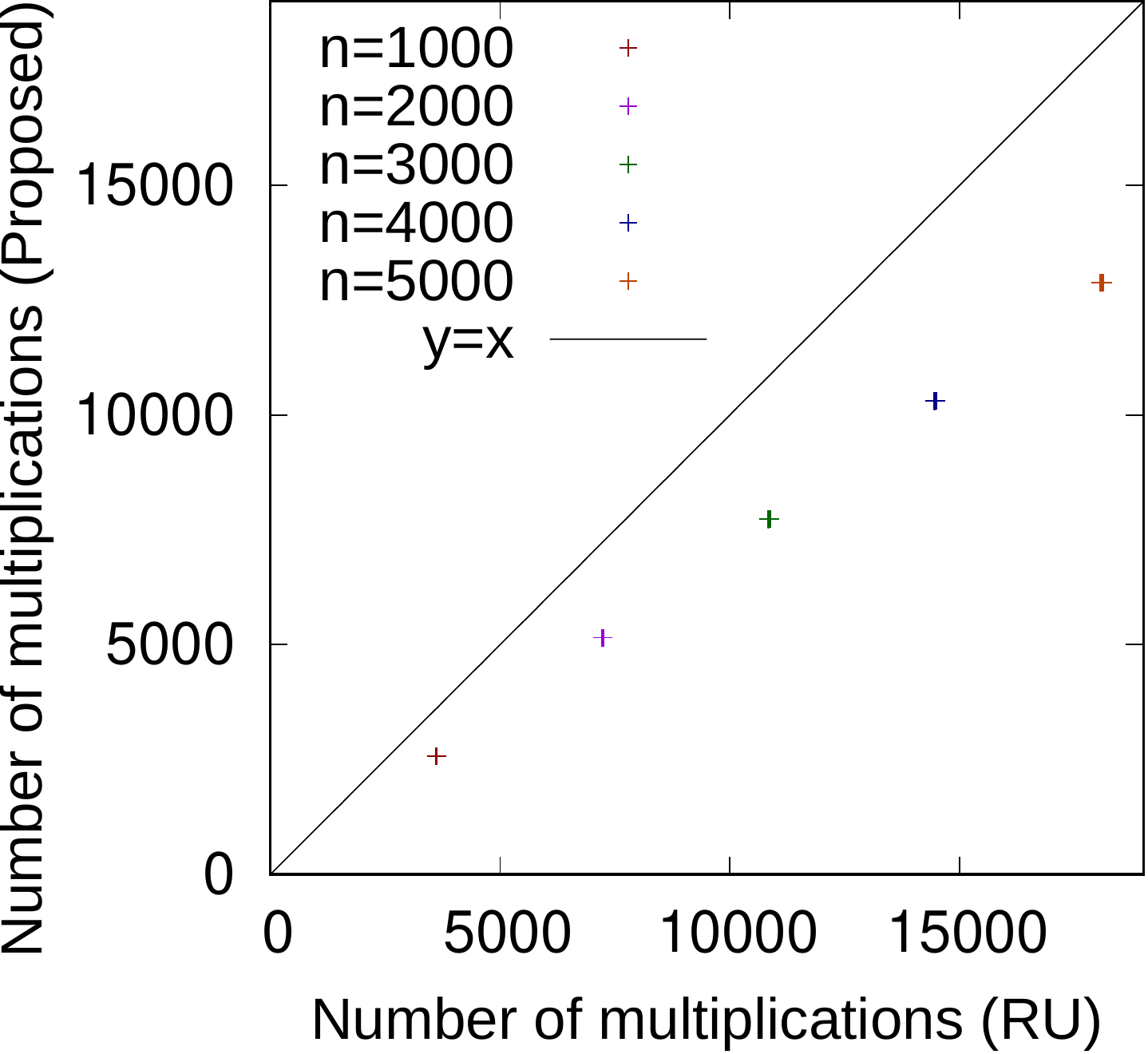} \label{fig:num1mr}
  }
  \subfigure[Additions for $\mathcal{E}_2$]{%
    \includegraphics[clip,width=.32\linewidth]{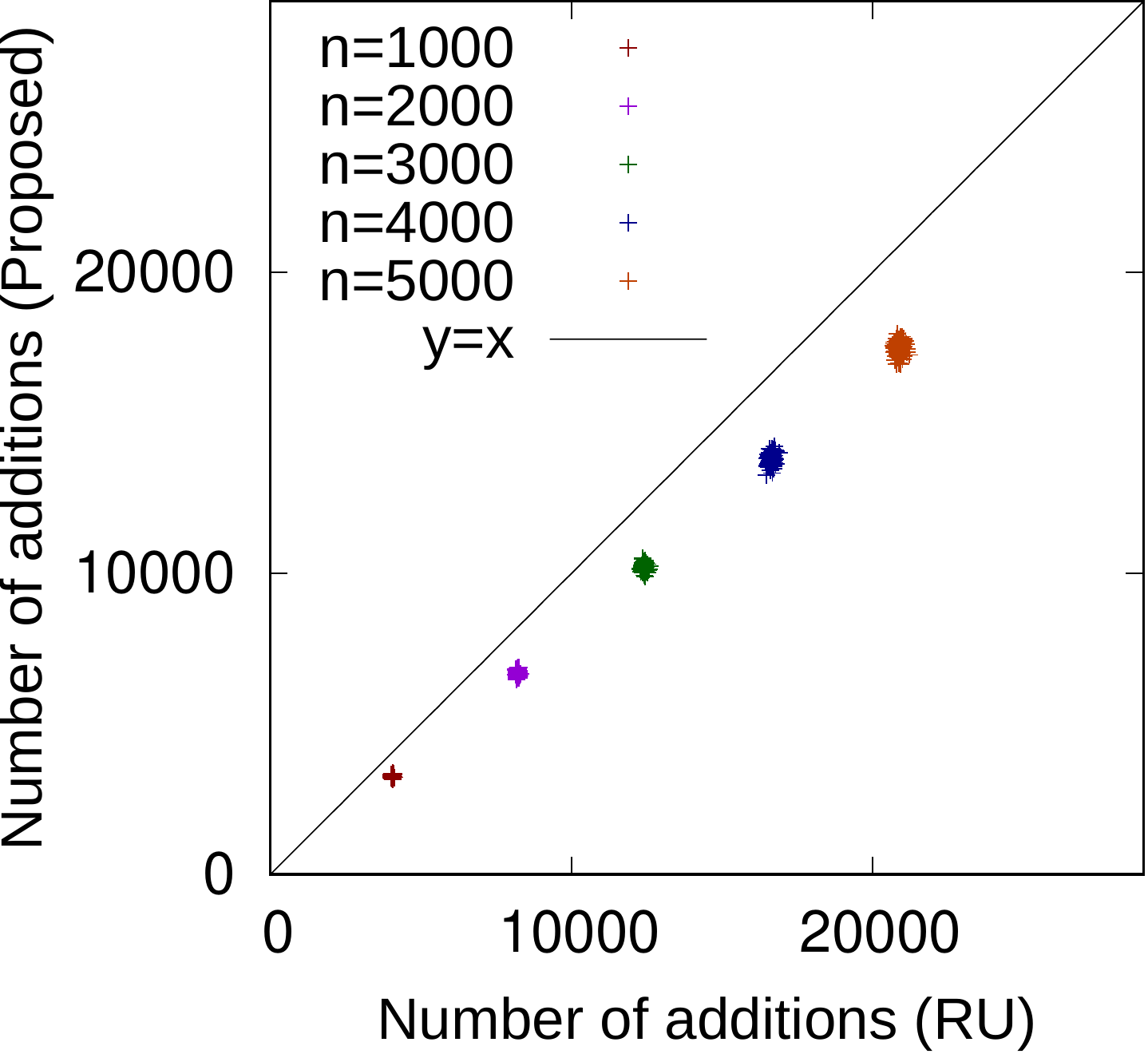} \label{fig:num2r}
  }
  \caption{Number of operations required in RU-EA and proposed EA \label{fig:num_r}}
\end{figure*}

\begin{figure*}[t]
  \centering
  \subfigure[Additions for $\mathcal{E}_8$]{%
    \includegraphics[clip,width=.32\linewidth]{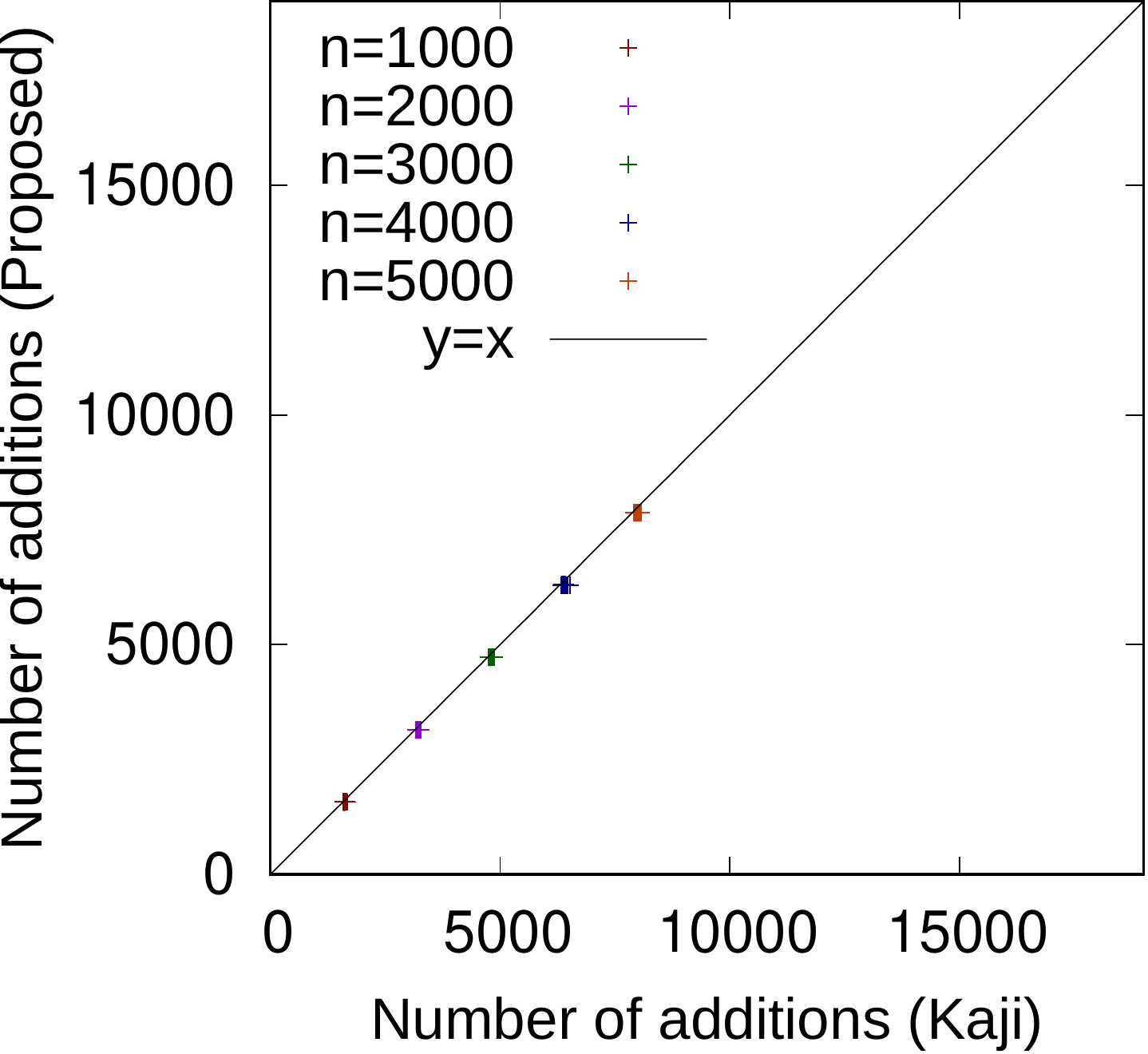} \label{fig:num1ak}
  }
  \subfigure[Multiplications for $\mathcal{E}_8$]{%
    \includegraphics[clip,width=.32\linewidth]{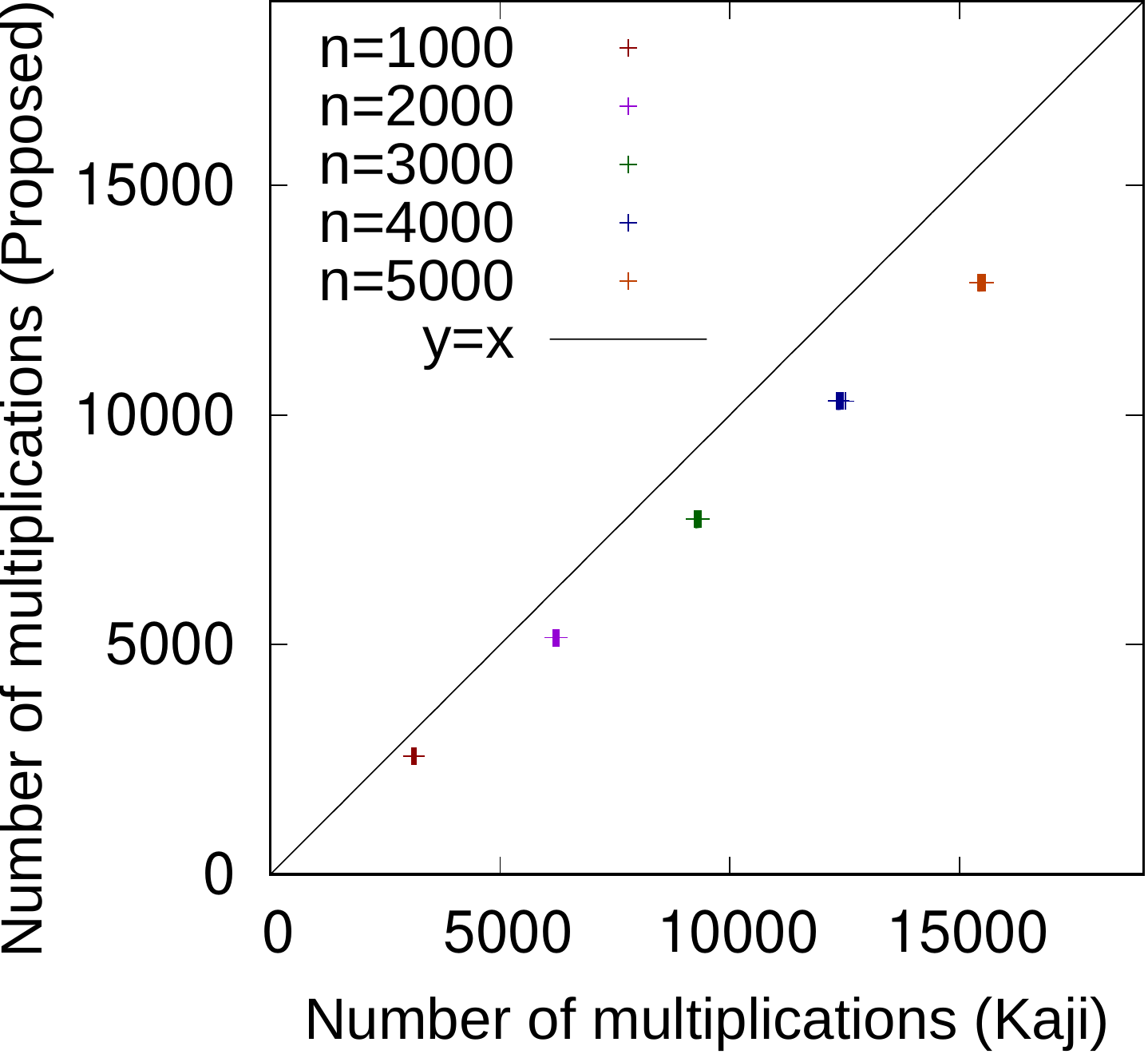} \label{fig:num1mk}
  }
  \subfigure[Additions for $\mathcal{E}_2$]{%
    \includegraphics[clip,width=.32\linewidth]{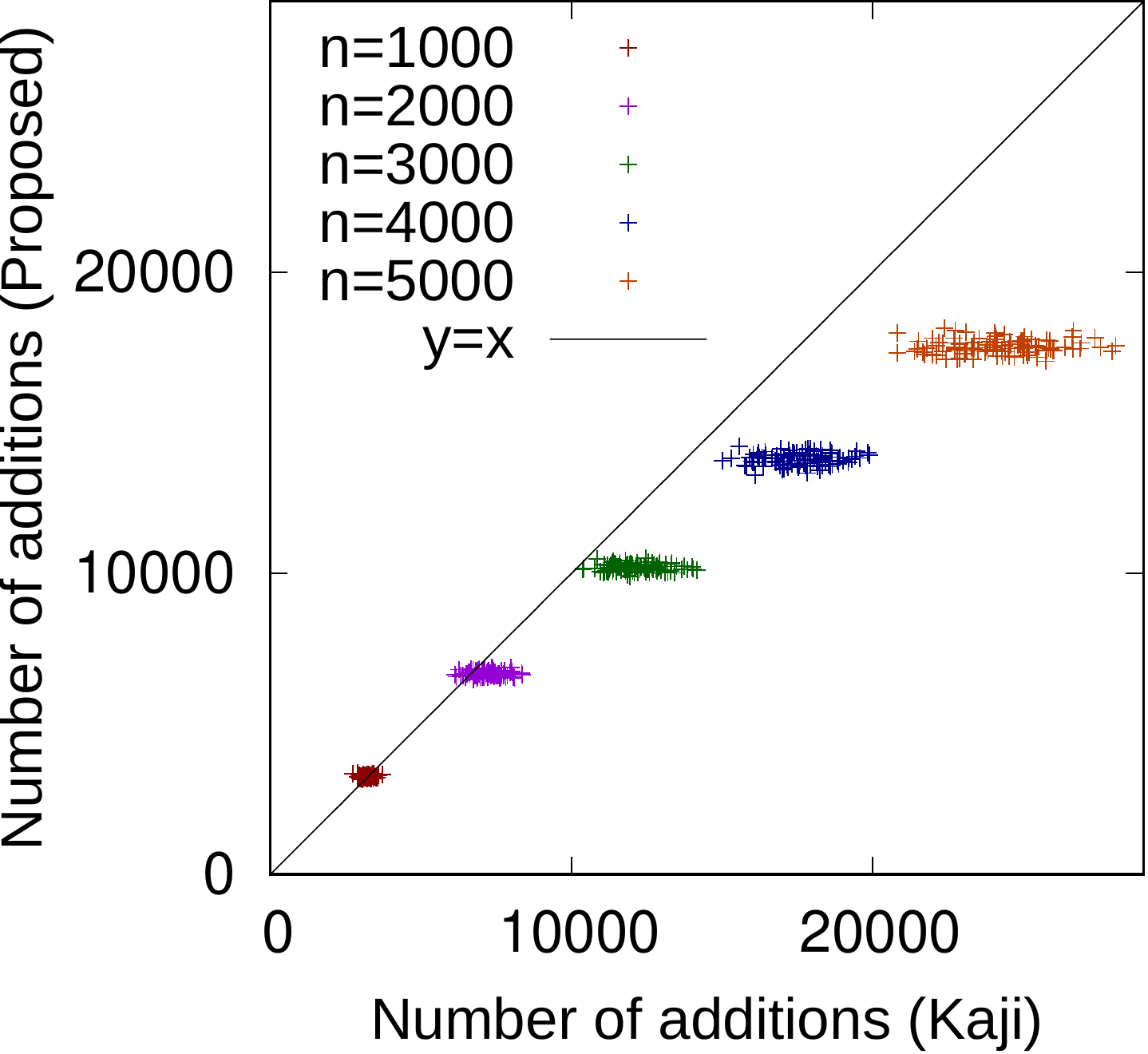} \label{fig:num2k}
  }
  \caption{Number of operations required in K-EA and proposed EA \label{fig:num_k}}
\end{figure*}

As a non-binary case, we consider the irregular LDPC code ensemble $\mathcal{E}_8$ over $\mathbb{F}_8$ with the degree distributions \cite{hu2005regular}:
\begin{align*}
  \lambda(x) \! &= \! 0.49978x \!+\! 0.17434x^2 \!+\! 0.29967x^3 \!+\! 0.02622x^4,\\
  \rho(x) \! &= \!0.81315x^4 \!+\! 0.18685x^5.
\end{align*}
For a binary case, we consider the irregular LDPC code ensemble $\mathcal{E}_2$ defined by the degree distributions \cite{amraoui2006asymptotic}:
\begin{align*}
  \lambda(x)&=0.0739196x + 0.657891x^2 + 0.268189x^{12},\\
  \rho(x)&=0.390753x^4 + 0.361589x^5 + 0.247658x^9.
\end{align*}

Table \ref{tab:ave_comp} evaluates the average number of additions and multiplications by the RU-EA, the K-EA, and the proposed EA for the ensemble $\mathcal{E}_8$ and $\mathcal{E}_2$ with code length $n=1000,2000,\dots,5000$.
For each experiment, we generate 100 parity-check matrices from ensemble $\mathcal{E}_8$ and $\mathcal{E}_2$.
Note that we regard the number of multiplications of those algorithms is 0 for the binary case.
Table \ref{tab:ave_comp} shows that the proposed EA has the smallest complexity for the non-binary case.
In this case, we can reduce the number of multiplications especially.
Table \ref{tab:ave_comp} shows the proposed EA has the smallest complexity for binary code with $n\ge 2000$.
In particular, we can reduce the encoding complexity for the binary codes with long code length.

To compare the number of operations for individual codes, we plot Figs \ref{fig:num_r} and \ref{fig:num_k}.
In those figures, the horizontal (resp.\ vertical) axis represents the number of operations for the existing EA, i.e., RU-EA or K-EA (resp.\ proposed EA).
For example, Fig.~\ref{fig:num2k} compares the number of additions by the K-EA and the proposed EA for 100 codes in $\mathcal{E}_2$.
In this figure, we plot a dot at $(x,y)$ if the K-EA requires $x$ additions and the proposed EA requires $y$ additions for a fixed code in $\mathcal{E}_2$.
Hence, we plot 100 dots for each code length.

Figures \ref{fig:num1ar}, \ref{fig:num1mr}, and \ref{fig:num2r} show that the proposed EA requires lower operations than the RU-EA for every code.
From Fig.~\ref{fig:num1ak}, the proposed EA and K-EA require about the same number of additions for the codes in $\mathcal{E}_8$.
Figure \ref{fig:num1mk} shows the proposed EA requires lower multiplications than the K-EA for every code in $\mathcal{E}_8$.
Figure \ref{fig:num2k} shows the proposed EA has lower complexity than the K-EA for the long binary codes.

\subsection{Modification of Proposed EA for Short Code Length}\label{sec:improve}

\begin{figure}[tb]
  \centering
  \includegraphics[width=.7\linewidth]{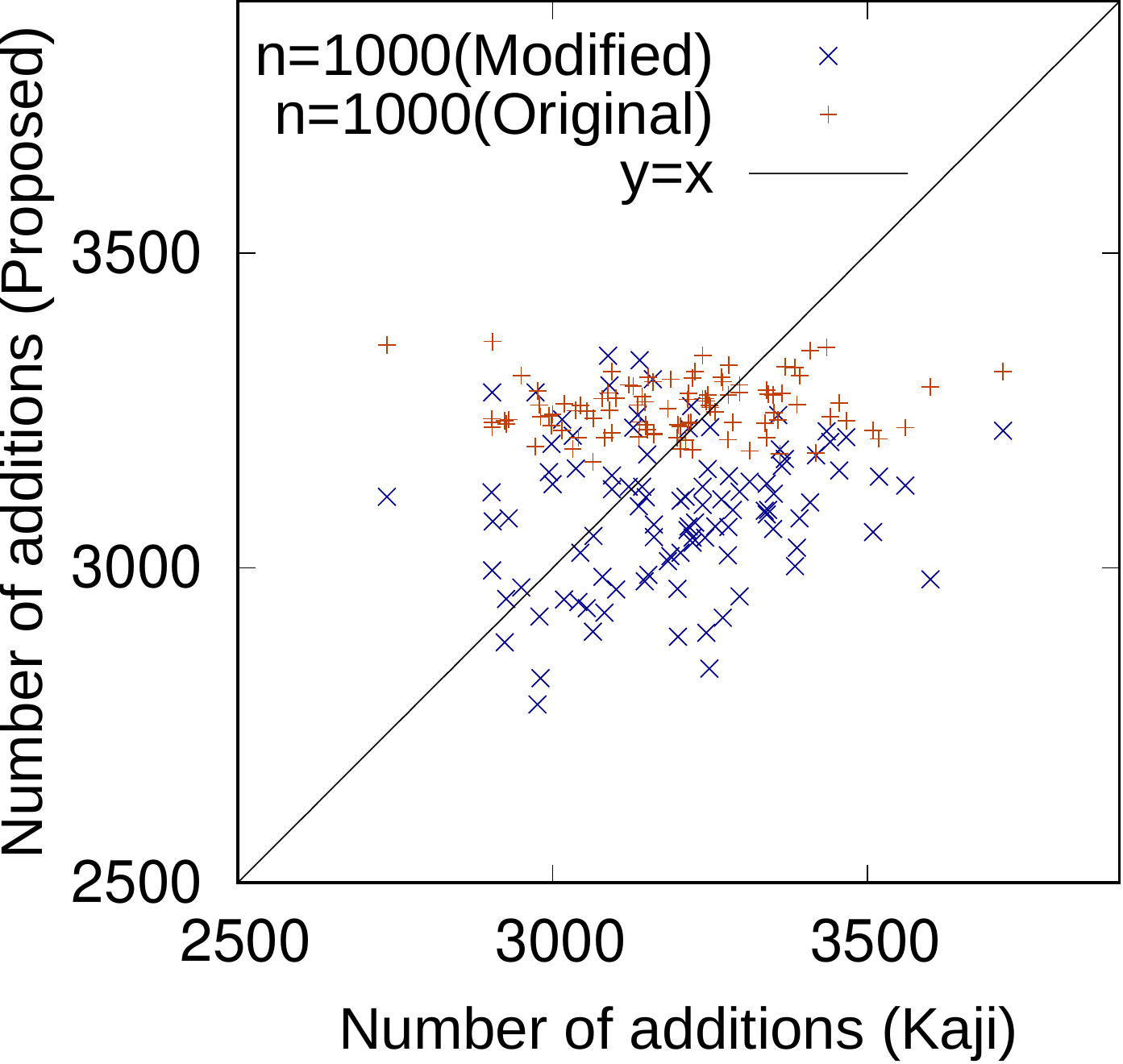}
  \caption{The number of additions to the K-EA and the proposed EA for codes in $\mathcal{E}_2$ with $n=1000$  \label{fig:mod} }
\end{figure}

As shown in the previous section, the proposed EA has a larger complexity than the K-EA for short binary LDPC codes.
As mentioned in Remark \ref{re:subm}, we can use another efficient algorithm to solve $\mat{A}_i \vec{p}_i^T = \vec{b}^T$.
In this section, we modify the proposed EA by the LU-factorization to reduce the complexity of the proposed EA for short binary LDPC codes.

We evaluate the number of additions required in the modified algorithm for 100 codes in $\mathcal{E}_2$ with $n=1000$.
Those results are plotted in Fig.~\ref{fig:mod}.
The average number of additions required in the modified EA is $3087.57$.
On the other hand, the average number of additions required in the K-EA is 3195.16 from Table \ref{tab:ave_comp}.
From the above, the modified EA can reduce the complexity and have lower average complexity than the K-EA.

\section{Conclusion}
In this paper, we have proposed an efficient EA for the binary and non-binary irregular LDPC codes.
As a result, the proposed EA has lower complexity than the RU-EA and the K-EA.

\section*{Acknowledgment}
We would like to thank Dr. Yuichi Kaji (Nagoya University) for providing the program of Kaji's EA.
This research was supported by JSPS KAKENHI Grant Number 16K16007 and Yamaguchi University Fund.

\end{document}